\RequirePackage{fix-cm}
\documentclass[twocolumn]{svjour3}          % twocolumn
\smartqed  % flush right qed marks, e.g. at end of proof
\usepackage{graphicx}
\usepackage{subfigure}
\usepackage{longtable}
\usepackage{graphicx,psfrag,epstopdf}
\usepackage{multirow}

\usepackage{algorithm}
\usepackage{algpseudocode}
\usepackage{pifont}

\usepackage{amsfonts}
\usepackage{amssymb}

\usepackage{float}

\usepackage{url}
\usepackage{amsmath}

\usepackage{csquotes}

\makeatletter
\def\BState{\State\hskip-\ALG@thistlm}
\makeatother
\begin{document}

\title{Layer Based Partition for Matrix Multiplication on\\ Heterogeneous Processor Platforms}




\author{Yang~Liu         \and
        Li~Shi          \and
        Junwei~Zhang    \and%etc.
        Thomas~G.~Robertazzi
}


\institute{Yang Liu \at
              \email{yangliu89415@gmail.com}           %  \\
           \and
           Li Shi \at
              \email{lishi.pub@gmail.com}
           \and
           Junwei Zhang \at
              \email{junwei.zhang@stonybrook.edu}
           \and
           Thomas~G.~Robertazzi \at
              ECE Department, Stony Brook University \\
              \email{Thomas.Robertazzi@stonybrook.edu }
}


\maketitle

\begin{abstract}
 While many approaches have been proposed to analyze the problem of matrix multiplication parallel computing, few of them address the problem on heterogeneous processor platforms. It still remains an open question on heterogeneous processor platforms to find the optimal schedule that balances the load within the heterogeneous processor set while minimizing the amount of communication. A great many studies are based on rectangular partition, whereas the optimality of rectangular partition as the basis has not been well justified.

In this paper, we propose a new method that schedules matrix multiplication on heterogeneous processor platforms with the mixed co-design goal of minimizing the total communication volume and the multiplication completion time. We first present the schema of our \emph{layer based partition} (LBP) method. Subsequently, we demonstrate that our approach guarantees minimal communication volume, which is smaller than what rectangular partition can reach. We further analyze the problem of minimizing the task completion time, with network topologies taken into account. We solve this problem in both single-neighbor network case and multi-neighbor network case. In single-neighbor network cases, we propose an equality based method to solve LBP, and simulation shows that the total communication volume is reduced by $75\%$ from the lower bound of rectangular partition. In multi-neighbor network cases, we formulate LBP as a Mixed Integer Programming problem, and reduce the total communication volume by $81\%$ through simulation. To summarize, this is a promising perspective of tackling matrix multiplication problems on heterogeneous processor platforms.

\end{abstract}

\section{Introduction}
\label{sec:introduction}
Matrix multiplication has been widely performed in a variety of areas. For example, in image processing, a multiplication of projection matrix and system coefficient matrix is used to reconstruct the original images from the projections \cite{GLZeng}. In signal processing, the discrete Fourier transform of a signal is calculated by multiplying the N-by-N DFT matrix with the signal matrix \cite{RGLyons}. Many other applications include cryptography, computer graphics, economics, physics, electronics, etc, which all involve large scale data processing.

On homogeneous processor platforms, the problem of scheduling matrix multiplication load for parallel processing has been extensively studied, such as Canon's algorithm \cite{Canon}, SUMMA \cite{Geijn} and Solomonik's 2.5D algorithm \cite{solomonik}, etc. However, these approaches generally ignore the heterogeneity of processors and links, as well as network topologies. Thus, in distributed computing and heterogeneous platforms, those algorithms fail to apply because they can't guarantee load balance in those scenarios. To minimize the task completion time without considering the heterogeneity of the system is simply impossible. Therefore, for distributed computing and heterogeneous systems, additional factors need to be taken into account, such as heterogeneous processor speeds, heterogeneous link speeds, network topologies, distributed storage, etc.

To schedule matrix multiplication on heterogeneous processor platforms, researchers usually consider the following questions.
\vspace{-0.8mm}
\begin{enumerate}
  \item How to allocate the computing load to minimize the total communication volume?
  \item How to minimize the task completion time?
\end{enumerate}
\vspace{-0.8mm}

To optimize the total communication volume, many previous research works apply \emph{rectangular partition} on the result matrix \cite{AKalinov}-\cite{Beaumont2}. Rectangular partition, which adopts the well-known \emph{divide and conquer} strategy, divides the result matrix into multiple sub-rectangles, and assigns each sub-rectangle's computing load to different processors respectively. However, approaches based on \emph{rectangular partition} generally have the following drawbacks:

\vspace{-0.8mm}
\begin{enumerate}
  \item The restriction of each division's shape being rectangular brings difficulty to find the optimal partition that minimizes the total communication volume.
  \item The best communication volume of \emph{rectangular partition} may not be globally optimal.
\end{enumerate}
\vspace{-0.8mm}

There have been perspectives using \emph{non-rectangular partition}\cite{DeFlumere1}\cite{DeFlumere2}. However, these approaches only allows one division of the partition to be in random shape, while keeping the majority of the rest still in the shape of a rectangle. Thus, these approaches do not completely resolve the problem brought by the \emph{rectangular partition}.

Motivated by this, we proposes another approach called \emph{layer based partition} (LBP). Rather than assigning one rectangular sub-matrix of the final result matrix to a specific processor to process, our algorithm assigns each processor with one layer. Each layer is of the same shape with the result matrix. We will show that this method guarantees the optimality of the total communication volume.

We further study the problem of scheduling matrix multiplication on heterogeneous processor networks with the goal of minimizing the multiplication completion time. While several approaches have been proposed \cite{Lastovetsky1}-\cite{Nagamochi}, none of them have addressed problem in the context of a specific network topology, like a heterogeneous mesh. Besides, most previous works consider the problem in the real number domain such that it is allowed for a processor to get $0.3$ rows, for instance. Comparatively, we consider the problem in the integer domain which is more applicable in real practice. In our paper, we study the problem in two specific networks: star network and mesh, and propose different strategies to minimize their overall finishing time.

\subsection{An Example}

\begin{figure}
\centering
  \includegraphics[width=3.4in]{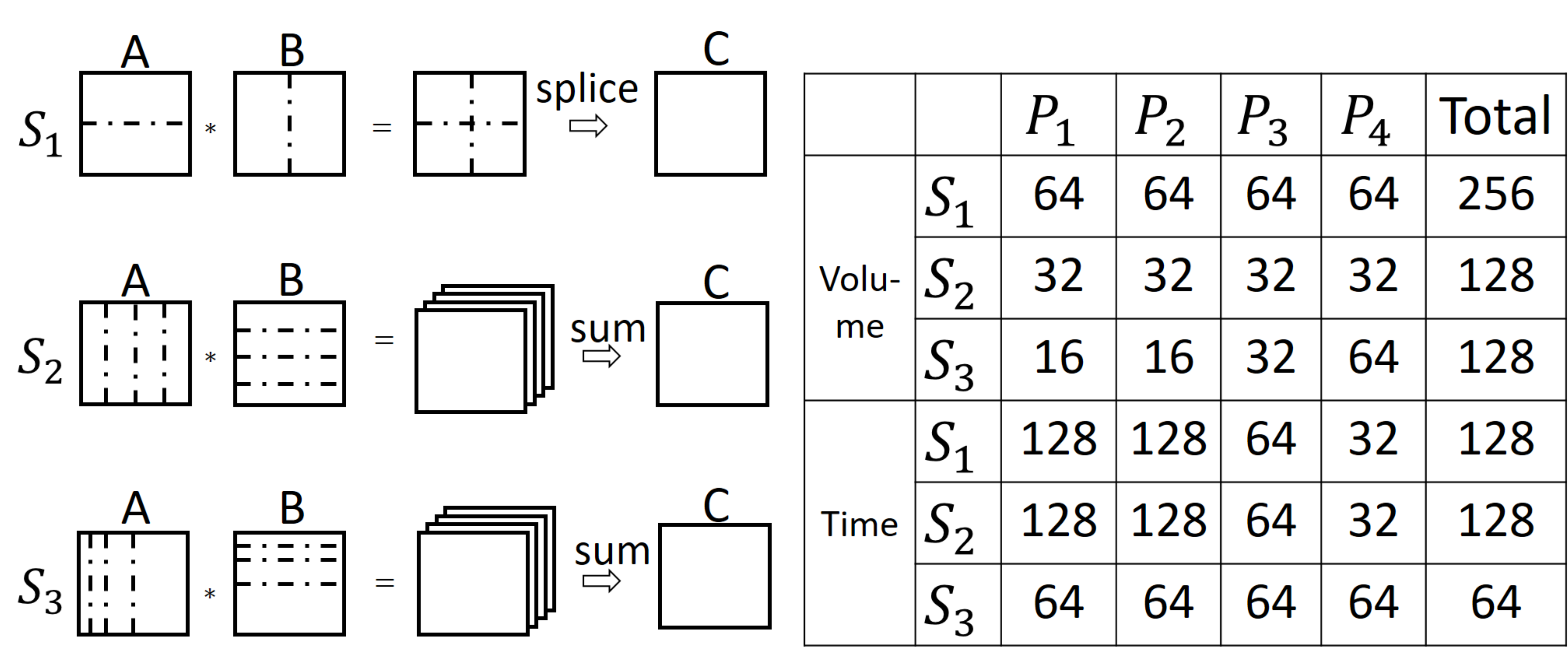}
  \caption{Examples: scheduling matrix multiplication load to 4 processors.}\label{fig:AnExample}
\end{figure}

To better illustrate this matrix multiplication scheduling problem, an example is shown in figure \ref{fig:AnExample}. Consider a task of multiplying two $8 * 8$ matrices(A and B) using four processors $P_{1}$, $P_{2}$, $P_{3}$ and $P_{4}$. Suppose the computing power of $P_{1}$, $P_{2}$, $P_{3}$ and $P_{4}$ are $1$, $1$, $2$ and $4$ respectively. How to schedule the computing load onto these four processors to optimize communication volume and multiplication finishing time? Figure \ref{fig:AnExample} provides three scheduling schemes $S_{1}$, $S_{2}$ and $S_{3}$. $S_{1}$ represents a \emph{rectangular partition} of the result matrix where each of the processor is assigned calculating one of the rectangular sub-matrices. $S_{2}$ and $S_{3}$ represent a \emph{layer based partition} scheme. Specifically, $S_{2}$ is an evenly divided scheme, in which each processor takes an equal number of rows and columns and computes one layer of the result matrix. Compared to $S_{1}$, $S_{2}$ remains unchanged in overall finishing time, but substantially reduces the total communication volume. $S_{3}$ also partitions the matrix multiplication load based on layers, but make the number of columns taken by each processor being proportional to its computing power. As a result, $S_{3}$ keeps a low communication volume like $S_{2}$, while optimizing the multiplication finishing time. Note that according to our previous assumption in the introduction part, we only consider the process of getting the distributed results, whereas aggregating the distributed results can be done asynchronously and is out of the scope of this discussion.

The reason that $S_{2}$ and $S_{3}$ has much less communication volume compared to $S_{1}$ is that entries of the matrices are sent only once in $S_{2}$ and $S_{3}$ whereas they are sent twice in $S_{1}$. For example, in $S_{1}$, the first row of matrix A is sent to both $P_{1}$ and $P_{2}$ in order to calculate the upper left sub-matrix and upper right sub-matrix respectively. In contrast, in $S_{2}$ and $S_{3}$, each entry of matrix $A$ and matrix $B$ is sent only once. Therefore, the total communication volume of $S_{2}$ and $S_{3}$ is greatly less than $S_{1}$.

We can see that when performing matrix multiplication, \emph{layer based partition} generates much less communication volume compared to \emph{rectangular partition}. With correct distribution of multiplication loads to processor, we can further optimize the multiplication completion time. However, in practice, the problem of scheduling matrix multiplication load is more sophisticated, when we consider this problem in a specific network environment, in which factors like network topology, communication mode, etc have to be taken into account.
It's even more complex, if we consider the heterogeneity of the system. This paper explores these cases in detail.

\subsection{Assumption and Constraints}
\label{sec:assumption}
In this paper we only consider acquiring the distributed results of each layer from the multiplication process, whereas the aggregation of those layers is out of the scope of this discussion. This is because we can exploit distributed storage to store those results in a distributed manner.
Since addition processes are of much lighter weight than multiplication processes, we can do asynchronous aggregation afterwards or only when necessary, rather than summing up immediately.
Therefore, we assume that we are at a decent state once all the $O(N^3)$ multiplication procedures are done. Then,
we store these multiplication results distributively in the memory or hard disk of each processor.
In the concept of \emph{rectangular partition}, that means we consider task completion once each sub-rectangle result is acquired, while the combination of these sub-rectangles are out of the scope of this discussion. In the concept of our \emph{layer based partition}, that means we regard task as completed once each layer of the result matrix has
been calculated and stored, whereas the summation process of these layers are left with asynchronous processes.
In either way, once all the distributed multiplication results are acquired, we mark that as the completion of the task.
This is an important assumption for the comparison between \emph{layer based partition}
and \emph{rectangular partition} in the latter part of this paper.
We apply this assumption when calculating the total communication volume and task completion time.

The memory limit can be a tight constraint for \emph{layer based partition}, when the size of the matrix gets enormously large and exceeds memory. For that case, we need to hold the matrices in hard disks, and batch-load the memory and do processing. After processing we write the result back to disk. With distributed storage we can store each layer of the result matrix on the hard disk of each processor.
Since these super-large matrices are usually not read intensive, we can do asynchronous sync-up, or only need to read from the disk and add up each layer when there's an infrequent read request.
All in all, it's a bit tricky for \emph{layer based partition} to handle matrices larger than memory size. But for other matrices, it's convenient and bears substantial advantages.

\subsection{Our Contributions}

Our main contributions are:

\vspace{-0.8mm}
\begin{enumerate}
  \item We analyze the disadvantages of \emph{rectangular partition.} First, we show it's essentially difficult to obtain a communication-optimal partition. Second, we show that the best communication volume achieved by \emph{rectangular partition} is not globally optimal. We propose the lower bound of matrix multiplication's communication volume, and we prove \emph{rectangular partition} consumes a total communication volume larger than the lower bound.
  \item In contrast, we propose \emph{layer based partition}(LBP) scheme. We show that our scheme is superior to \emph{rectangular partition} because: A). It is easy to obtain a communication-optimal partition. B). It can reach the lower bound of total communication volume.
  \item We study the load scheduling problem to minimize overall finishing time in star networks. We propose an equality based strategy, and apply this strategy under four communication modes respectively: SCSS, SCCS, PCCS, PCSS.
  \item We study the load scheduling problem to minimize overall finishing time in mesh networks. We formulate this problem as a mixed integer programming problem called \emph{MFT-LBP}. We propose an algorithms called \emph{PMFT-LBP} to solve it, which contains 3 phases. We also provide a heuristic to reduce time complexity.
\end{enumerate}
\vspace{-0.8mm}

One thing to notice is that the matrices discussed in this paper are dense matrices with hundreds or thousands of rows/columns, which are very
common in current parallel processing environment.

The rest of the paper is organized as follows. Section $2$ reviews the related research. Section $3$ introduces \emph{layer based partition} scheme, and analyzes its improvements over \emph{rectangular partition}. Section $4$ studies the load scheduling scheme of LBP in $single-neighbor$ network, and Section $5$ studies that in $multi-neighbor$ network. Section $6$ evaluates the performance of our algorithms through simulations. Finally, section $7$ concludes the paper.

\section{Related Work}
\label{sec:related}
A great amount of research effort has been devoted to the problem of matrix multiplication parallel/distributed processing. In this section, we categorize those most related works into the following parts:

\smallskip \noindent {\bf Approaches on homogeneous platforms.} Homogeneous platforms assume that all the computing / comunication resources and environment are identical. Matrix multiplication scheduling on homogeneous platforms have been extensively studied in \cite{Canon} - \cite{Samantray}, among which, Canon introduces the first parallel algorithm on homogeneous grids \cite{Canon}. Fox {\em et al.} extend the analysis on two-dimensional mesh and hypercubes \cite{Fox}. SUMMA overcomes the shortcomings of Cannon's and Fox's, and becomes the most widely applied parallel matrix multiplication scheme \cite{Geijn}.
Solomonik {\em et al.} \cite{solomonik} propose a method known as the `2.5D Algorithm', which can achieve asymptotically less communication than Canon's algorithm and be faster in practise.
Malysiak {\em et al.} \cite{Malysiak} present a novel model of distributing matrix multiplication within homogeneous systems with multiple hierarchies. Scheduling of sparse-dense matrix multiplication has been studied in \cite{Yzelman}\cite{Koanantakool}.

However, when the computing scale gets larger and larger, heterogeneous computing is more suitable than homogeneous computing for super large distributed and parallel processing. An example is, CPU-GPU heterogeneous processing is inevitable a trend to achieve higher computing performance \cite{Mittal}. Heterogeneous System Architecture(HSA) is another example that uses multiple processor types on the same integrated circuit to provide the best overall performance \cite{Hwu}.

\smallskip \noindent {\bf Approaches on heterogeneous platforms.} Efforts have been devoted to analyze matrix multiplication on heterogeneous platforms \cite{AKalinov}-\cite{Yang}. Some researchers try to extend those parallel processing algorithms from homogeneous platforms to heterogeneous platforms. Both Kalinov \cite{AKalinov} and Quintin {\em et al.} \cite{Quintin} investigate the scalability to modify SUMMA \cite{Geijn}, to fit into heterogeneous processor platforms. Ohtaki {\em et al.} \cite{Ohtaki} propose a scheme to apply the Strassen's algorithm on heterogeneous clusters.

However, these approaches only optimize communication volume, yet fail to consider whether multiplication completion time gets optimized as well.

In addition to approaches that try to apply homogeneous parallel algorithms, other researchers seek new perspectives. Alonso {\em et al.} \cite{Alonso} use two strategies to implement parallel solvers for dense linear algebra problems on heterogeneous clusters. Malik {\em et al.} \cite{Malik} propose a topology-aware matrix multiplication algorithm, and they base their hierarchical communication model on irregular 2D rectangular partition. Zhong {\em et al.} \cite{Zhong} utilize functional performance model to balance load on heterogeneous networks of uni-processor computers. Demmel {\em et al.} \cite{Demmel} propose a communication-efficient algorithm for all dimensions of rectangular matrices, apart from square matrices. Beaumont {\em et al.} \cite{Beaumont1} compares static, dynamic and hybrid resource allocation strategies for matrix multiplication, and analyse the benefit of introducing more static knowledge in runtime libraries.

The most utilized partition method in these approaches is \emph{rectangular partition}. There exists a significant amount of research digging into the problem of \emph{rectangular partition} on heterogeneous platforms, as discussed below.

\smallskip \noindent {\bf Rectangular partition approaches on heterogeneous platforms.} Many approaches have been proposed based on \emph{rectangular partition} of matrix, and researchers have considered the optimal \emph{rectangular partition} \cite{Lastovetsky1}\cite{Lastovetsky2}\cite{Clarke}. However, even though the optimal \emph{rectangular partition} problem attracts a great deal of attention, the optimal partition that minimizes the total communication volume remains an open question. Researchers have realized that the problem is complex. Ballard {\em et al.} \cite{Ballard} study the lower bounds of communication volume, and the difficulty of finding communication optimal \emph{rectangular partition}. Beaumont {\em et al.} \cite{Beaumont2} prove that given the area of each sub-rectangle, it is a NP-complete problem to find communication optimal \emph{rectangular partition} of a specific matrix. In regard to this, DeFlumere {\em et al.} \cite{DeFlumere1} question the optimality of \emph{rectangular partition}, and propose a perspective of \emph{non-rectangular partition}. Further, DeFlumere {\em et al.} \cite{DeFlumere2} show that this non-rectangular scheme outperforms \emph{rectangular partition} in terms of communication volume, and generalize this algorithm to a three processors' scenario. Nagamochi {\em et al.} \cite{Nagamochi} propose a recursive partitioning algorithm that dissects a rectangle into rectangles with specified areas. Beaumont {\em et al.} \cite{Beaumont3} in addition, propose a new approximation algorithm for matrix partitioning by adopting the idea of non-rectangular partition, and the recursive partitioning algorithm proposed by Nagamochi {\em et al.} Fuenschuh {\em et al.} \cite{Fuenschuh} present a polynomial time approximation algorithm that solves a soft rectangle packing problem, and derive an upper bound estimation on its approximation ratio.

In summary, researchers have begun to realize the difficulty in finding the optimal \emph{rectangular partition} that balance loads while minimize communication volume. Beaumont's work {\em et al.} \cite{Beaumont2} explicitly reveals that it is a NP-complete problem. Moreover, though alternative perspectives like \emph{non-rectangular partition} have been proposed \cite{DeFlumere1}\cite{DeFlumere2}, those perspectives keep the majority of the partitions still in the shape of rectangle, which do not completely resolve the restriction brought by geometrical shapes. In contrast, in \emph{layer based partition} scheme, we avoid the NP-complete problem and make it very easy to obtain a communication-optimal partition, which reaches the lower bound of total communication volume.

\section{Layer Based Partition(LBP)}\label{layer partition}
In this section we propose a new method - the \emph{layer based partition}(LBP) scheme to tackle matrix multiplication on heterogeneous processor platforms.

\begin{figure}[H]
\centering
  \includegraphics[width=3.5in]{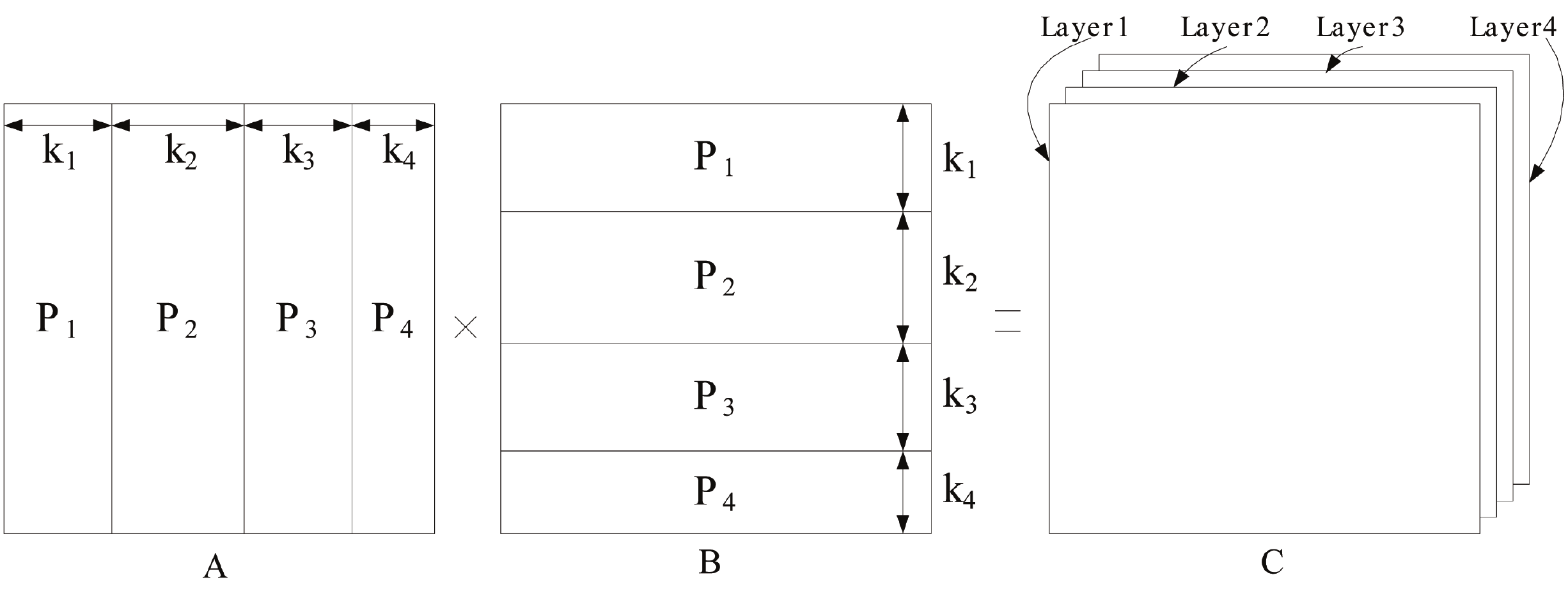}
  \caption{Layer Based Partition Pattern.}\label{fig:Layer}
%\vspace{-0.2in}
\end{figure}

\subsection{Scheme Overview}
While the goal remains as conducting two $N * N$ matrices' multiplication, the approach taken by LBP is different from \emph{rectangular partition}. In LBP, each processor is responsible for calculating one layer of the output matrix. Each layer is of the same dimension of the output matrix, and the output matrix is the aggregation of all layers.

Figure \ref{fig:Layer} shows an example of this LBP scheme with four processors cooperating on two $N * N$ matrices' multiplication. Processor $P_{1}$ takes matrix $A$'s leftmost $k_{1}$ columns and $B$'s upmost $k_{1}$ rows, and then do multiplication. The result is still a $N * N$ matrix, which is actually the 1st layer of the final output matrix. The same method applies to the rest of the processors, with processor $P_{2}$ taking charge of 2nd layer, processor $P_{3}$ taking charge of 3rd layer, and processor $P_{4}$ taking charge of the last layer. The final output matrix is the sum of these four layers.

\subsection{Improvements over Rectangular Partition}
The important improvements of LBP scheme over \emph{rectangular partition} include
\begin{itemize}
\item It is much easier for LBP to obtain a communication-optimal partition. In fact, any \emph{layer based partition} is already communication-optimal. But for \emph{rectangular partition}, it is a NP-complete problem to find communication-optimal partition according to Beaumont \cite{Beaumont2}, because of the restriction that each division's shape has to be rectangular.

\item Under assumption \ref{sec:assumption}, LBP also takes less total volume of data sent by the source than \emph{rectangular partition}. LBP essentially reaches the communication lower bound, as will be proven below.
\end{itemize}

In order to ensure an equal base of comparison, we assume in the following part of this paper that, all source nodes do not take part in computation.
\begin{theorem}
\emph{(Layer Based Partition Theorem)}
\label{LayerBasedPartitionTheorem}
LBP generates the minimal total communication volume in conducting two square matrices' multiplication.
\end{theorem}

To prove the theorem, we firstly present the lower bound of communication volume.

\begin{lemma}
\label{LowerBound}
For all scheduling schemes for two $N * N$ matrices' multiplication, the lower bound of communication volume is $2N^{2}$, if the source does not take part in processing.
\end{lemma}

\begin{proof}
    When conducting two $N * N$ matrices' multiplication, both matrices have to be sent from the source to the computing node, because the source doesn't take part in computing. Each matrix contributes a communication volume of $N^{2}$, so two matrices together are $2N^{2}$. Since each entry of these two matrices has to be sent at least once, the total communication volume is always greater than or equal to $2N^{2}$. Thus the lower bound stands.
\end{proof}

Back to Theorem \ref{LayerBasedPartitionTheorem}. Suppose there are $p$ processors, therefore according to LBP, the task is divided into $p$ layers. Processor $p_{1}$ takes matrix $A$'s leftmost $k_{1}$ columns and $B$'s upmost $k_{1}$ rows, and thus the communication volume to transfer the necessary processing data from source to processor $p_{1}$ is $N * k_{1} + N * k_{1} = 2Nk_{1}$. Similarly, the communication volume of processor $p_{2}$ is $2Nk_{2}$, processor $p_{3}$ is $2Nk_{3}$, ...etc. The total communication volume is \[C_{LBP} = \sum_{i=1}^{p}{2Nk_{i}} = 2N\sum_{i=1}^{p}{k_{i}} = 2N^2\]

As shown above, the communication volume of \emph{layer based partition} scheme reaches the lower bound. Therefore, \emph{layer based partition} scheme is communication-optimal.

\begin{lemma}
\label{RectangularPartitionCommunicationVolume}
Rectangular partition consumes a total communication volume greater than the lower bound.
\end{lemma}

\begin{proof}
    In \emph{rectangular partition}, suppose the output matrix is of size $N * N$, and each sub-rectangle's height is $h_{i}$, width is $w_{i}$, and area is $s_{i}$. The total communication volume according to \cite{Beaumont2} is: \[C_{REC} = \sum_{i=1}^p{(h_{i} + w_{i}) * N}\]
    Because $h_{i} + v_{i} \geq 2\sqrt{s_{i}}$, therefore we have:
    \begin{equation}
    \label{equ:commVolInequ}
        C_{REC} \geq 2N * \sum_{i=1}^p{\sqrt{s_{i}}}
    \end{equation}
    In the meanwhile, we have the sum of each sub-rectangle: \[\sum_{i=1}^p{s_{i}} = N^2\]
    \begin{equation}
    \label{equ:sumRecArea}
        \sum_{i=1}^p{(\sqrt{s_{i}})^2} = N^2
    \end{equation}
    Since each $s_{i}$ is positive and $p > 1$: \[\sum_{i=1}^p{(\sqrt{s_{i}})^2} < (\sum_{i=1}^p{\sqrt{s_{i}}})^2\]
    Combined with equation (\ref{equ:sumRecArea}), we get: \[(\sum_{i=1}^p{\sqrt{s_{i}}})^2 > N^2\]
    \begin{equation}
    \label{equ:sumSqrtRecArea}
        \sum_{i=1}^p{\sqrt{s_{i}}} > N
    \end{equation}
    Take (\ref{equ:sumSqrtRecArea}) back to (\ref{equ:commVolInequ}),
    \begin{equation}
        C_{REC} > 2N^2
    \end{equation}
\end{proof}

Thus the total communication volume of LBP($C_{LBP}$) is less than \emph{rectangular partition}($C_{REC}$).

\subsection{Memory Limit Constraint}
\label{sec:memory}
The memory limit can be a tight constraint for \emph{layer based partition}, for the case where the result matrix's size is large enough to exceed memory.
That's a bottleneck. For that case, because we can't hold the entire result matrix in memory, we need to utilize hard disk and batch processing. We can load the memory of each processor to its limit, do the multiplication and write the result to disk. After that, we continue to load the next batch of data until completing all data sets. This approach is doable with the fast development of hardware like SSD and et cetera.

After the result of one layer is acquired on a single machine, we can leave the result there on the hard disk of that machine.
We can do asynchronous processing to aggregate each layer. Alternatively, we can do lazy sync-up.
For lazy sync-up we only need to read from the disks and aggregate each layer when
there¡¯s a read request, which shouldn't be very frequent for large matrices like this.
After all, if a result matrix's size exceeds the memory limit, after aggregating all the layers,
it has to be stored on hard disks anyway. We change this centralized storage manner to a distributed manner, by storing each layer of the result matrix on the hard disk of each processor.

\subsection{Summary}
To summarize, \emph{layer based partition} scheme avoids the NP-complete \emph{rectangular partition} problem, and has been proven to be communication-optimal. The next step is to discuss how to apply it in different network topologies.

In the following section, we discuss the details of applying the LBP scheme on heterogeneous processor networks. We focus on analyzing LBP's load balancing strategy, namely, how to schedule load distribution among different processors in the network. Since the LBP scheme already ensures communication volume to be minimal, the optimization goal of the load balancing strategy is to optimize the task completion time, which is another primary and widely discussed target of optimization on heterogeneous processor platforms. Again, as we mentioned in the introduction part, we mark it as the completion of the task once all the distributed multiplication
results are acquired on each processor. The total time it takes to compute those multiplication results of all the layers is defined as the task completion time here.

\section{LBP in Single-Neighbor Networks}
\label{chap:LBPSingleNeighbor}
To begin with, we discuss applying LBP in networks where each node can only receive loads from at most one of its neighbors. That node is allowed to further send loads to its other neighbors, except for the one where it obtains data from. We name this type of networks \emph{single-neighbor networks} in the following parts of our paper. Typical examples of \emph{single-neighbor networks} include star network, tree network, muti-level tree, etc. Here we use star networks to discuss the load balancing strategy under four scenarios: Sequential Communication Simultaneous Start (SCSS), Sequential Communication Consecutive Start (SCCS), Parallel Communication Simultaneous Start (PCSS), Parallel Communication Consecutive Start (PCCS). Sequential Communication means each node can only send load to its neighbors sequentially, while, Parallel Communication allows the source to transmit simultaneously. Consecutive Start means each non-source processor can only start processing data after receiving all the data it needs, while, Simultaneous Start allows each node to start processing while receiving data.

\begin{theorem}
\emph{(Single Source Network Load Balancing Theorem)}
\label{SingleNeighborNetworkLoadBalancingTheorem}
For single-neighbor networks, the condition to reach load balancing is to have all the nodes finish working at the same time.
\end{theorem}

This load balancing theorem comes from Bharadwaj etc.'s monograph \cite{Bharadwaj}. The idea is if not all processors finish processing at the same time, loads can be reduced from busy processors, and assigned to idle processors to speed up the overall process. Consequently, the minimal finishing time is obtained only when all processors finishing working at the same time.

\begin{table}[H]
\centering
\caption{Table of the Variables/Constants for Chapter \ref{chap:LBPSingleNeighbor}}
\begin{tabular}{l|l}
\hline
Var/Const & Meaning\\
\hline
$k_{i}$ & The number of rows or columns \\ & assigned to $i$th processor\\
$w_{i}$ & The inverse of the computing speed of \\ & $i$th processor\\
$z_{i}$ & The inverse of the link speed of the $i$th link\\
$T_{cp}$ & Computing intensity constant: \\
 & Unit load on $i$th processor is processed \\ & in $w_{i}T_{cp}$ seconds\\
$T_{cm}$ & Communication intensity constant: \\
 & Unit load on $i$th link can be transmitted \\ & in $z_{i}T_{cm}$ seconds\\
$T_{f}$ & The total finishing time of the entire network\\
\hline
\end{tabular}
\label{table:variablesSingleNeighbor}
\vspace{-0.15in}
\end{table}

In practice, because $k_{i}$ must be an integer, so it may be impossible to generate a set of integers $\{k_{i}\}$ that makes each processor finish exactly at the same time. Our approach is to relax each $k_{i}$ as a real number first, solve the relaxed problem, then find the integer solution that is closest to the solution of the relaxed problem. We believe that by making the integer solution as close to the optimal solution of the relaxed problem (in real domain) as possible, it is more likely that we can get the optimal solution of the original problem in which each $k_{i}$ is an integer.

\subsection{Sequential Communication Simultaneous Start}
In this subsection, we analyze the SCSS case where the source transmit loads to each processor sequentially, and in the mean time communication can overlap with computing. Figure \ref{fig:SCSSandSCCS}(a) displays the time sequence of this case.

\begin{figure}[H]
\centering
  \includegraphics[width=3.6in]{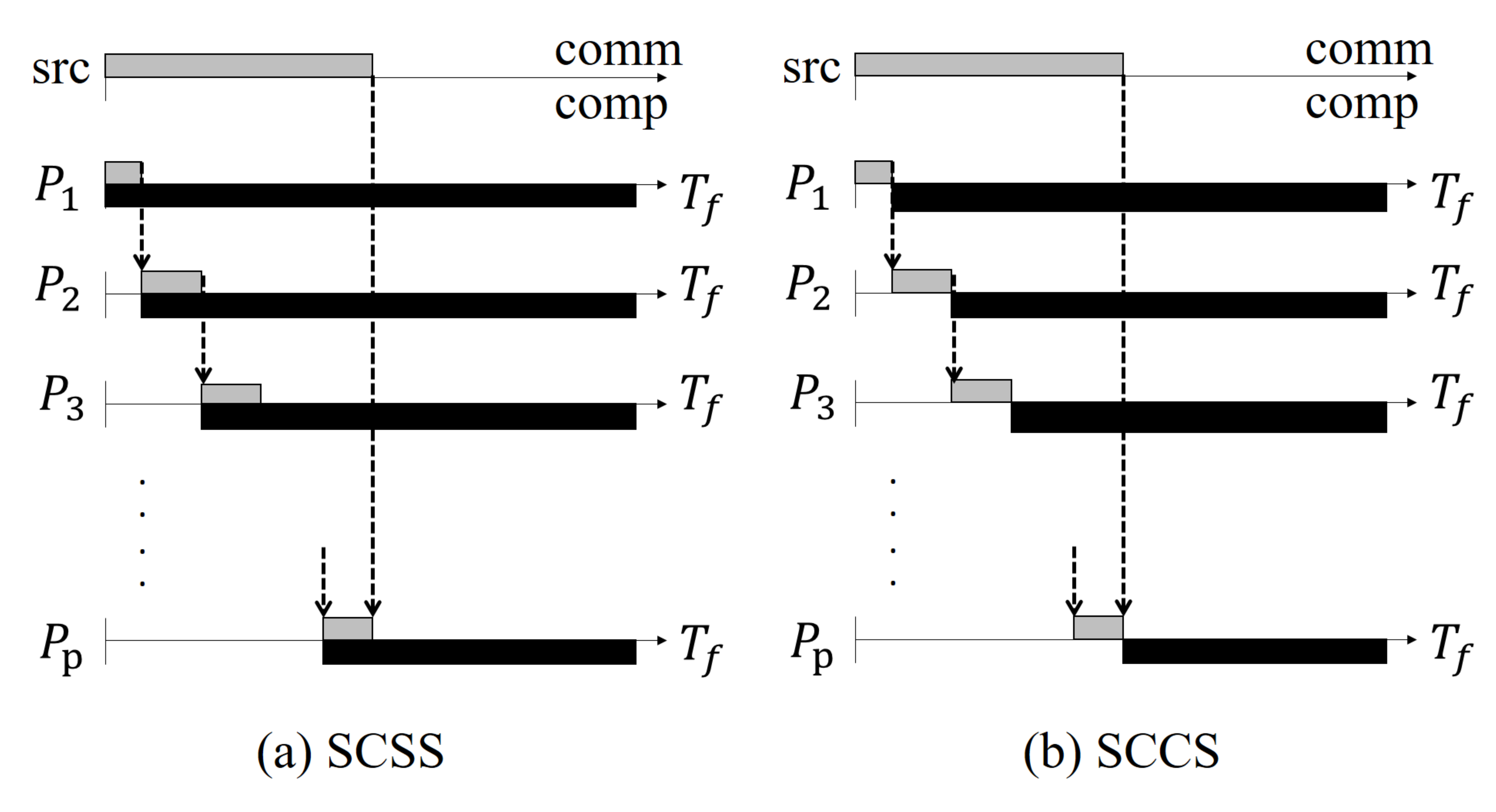}
  \caption{SCSS mode and SCCS mode.}\label{fig:SCSSandSCCS}
\end{figure}

Each processor calculates one layer of the result matrix. Take processor $i$ for instance. To work out $i^{th}$ layer's data, processor $i$ needs to get $N$ rows $* ~k_{i}$ columns of $A$'s data, plus $k_{i}$ rows $* ~N$ columns of $B$'s data. Thus its communication volume is $2k_{i}N$, and the corresponding communication time on $i^{th}$ link is $2k_{i}Nz_{i}T_{cm}$. Furthermore, for each entry of $i^{th}$ layer, it needs $k_{i}$ multiplications to get its value. There are as many as $N^{2}$ entries in this layer. Therefore the total number of multiplication is $k_{i}*N^{2}$. The corresponding computation time is $k_{i}N^{2}w_{i}T_{cp}$.
Consider all the timing relationship shown in Figure \ref{fig:SCSSandSCCS}(a), we have the following equations:

\begin{equation}
k_{1}N^{2}w_{1}T_{cp} = k_{2}N^{2}w_{2}T_{cp} + 2k_{1}Nz_{1}T_{cm}
\end{equation}
\begin{equation}
k_{2}N^{2}w_{2}T_{cp} = k_{3}N^{2}w_{3}T_{cp} + 2k_{2}Nz_{2}T_{cm}
\end{equation}
\begin{equation}
k_{i-1}N^{2}w_{i-1}T_{cp} = k_{i}N^{2}w_{i}T_{cp} + 2k_{i-1}Nz_{i-1}T_{cm}
\end{equation}
\[
\begin{array}{c}
\vdots
\end{array}
\]
\begin{equation}
k_{p-1}N^{2}w_{p-1}T_{cp} = k_{p}N^{2}w_{p}T_{cp} + 2k_{p-1}Nz_{p-1}T_{cm}
\end{equation}
\begin{equation}
k_{1}+k_{2}+k_{3}+\cdots+k_{p} = N
\end{equation}

The set of equations above consists of $p-1$ equations specifying the time sequence relationship between pairwise processors, and one equation specifying the normalization constraint. Meanwhile, we have $p$ unknown variables $k_{1}, k_{2}, ..., k_{p}$. Therefore, the set of equations are solvable. Solving them, we get:
\begin{equation}
k_{i} = \prod_{j=2}^{i}\frac{Nw_{j-1}T_{cp}-2z_{j-1}T_{cm}}{Nw_{j}T_{cp}}k_{1}, i=2,3,...p
\end{equation}

where $k_{1}$ is defined as:
\begin{equation}
k_{1} = \frac{N}{1 + \sum_{i=2}^{p}\prod_{j=2}^{i}\frac{Nw_{j-1}T_{cp}-2Z_{j-1}T_{cm}}{Nw_{j}T_{cp}}}
\end{equation}

The overall finishing time of the whole network can be obtained through the following equation.
\begin{equation}
T_{f} = k_{1}N^{2}w_{1}T_{cp}
\end{equation}

\subsection{Sequential Communication Consecutive Start}
In this section we discuss the case where the source sequentially assigns load to each processor, and communication can not overlap with computation. According to Figure \ref{fig:SCSSandSCCS}(b), we have the following equations:

\begin{equation}
k_{1}N^{2}w_{1}T_{cp} = k_{2}N^{2}w_{2}T_{cp} + 2k_{2}Nz_{2}T_{cm}
\end{equation}
\begin{equation}
k_{2}N^{2}w_{2}T_{cp} = k_{3}N^{2}w_{3}T_{cp} + 2k_{3}Nz_{3}T_{cm}
\end{equation}
\begin{equation}
k_{i}N^{2}w_{i}T_{cp} = k_{i+1}N^{2}w_{i+1}T_{cp} + 2k_{i+1}Nz_{i+1}T_{cm}
\end{equation}
\[
\begin{array}{c}
\vdots
\end{array}
\]
\begin{equation}
k_{p-1}N^{2}w_{p-1}T_{cp} = k_{p}N^{2}w_{p}T_{cp} + 2k_{p}Nz_{p}T_{cm}
\end{equation}
\begin{equation}
k_{1}+k_{2}+k_{3}+\cdots+k_{p} = N
\end{equation}

These $p$ equations contains $k_{1}, k_{2}, ..., k_{p}$ unknown variables. The solution set is:
\begin{equation}
k_{i} = \prod_{j=2}^{i}\frac{Nw_{j-1}T_{cp}}{Nw_{j}T_{cp} + 2z_{j}T_{cm}}k_{1}, i=2,3,...p
\end{equation}

where $k_{1}$ is defined as:

\begin{equation}
k_{1} = \frac{N}{1+\sum_{i=2}^{p}\prod_{j=2}^{i}\frac{Nw_{j-1}T_{cp}}{Nw_{j}T_{cp} + 2z_{j}T_{cm}}}
\end{equation}

The overall finishing time of the whole network is:
\begin{equation}
T_{f} = k_{1}N^{2}w_{1}T_{cp} + 2k_{1}Nz_{1}T_{cm}
\end{equation}

\subsection{Parallel Communication Consecutive Start}

For the case where source communicates with each processor in a parallel manner, and communication can not overlap with computation, we have the following equations:
\begin{equation}
k_{1}N^{2}w_{1}T_{cp} + 2k_{1}Nz_{1}T_{cm} = k_{2}N^{2}w_{2}T_{cp} + 2k_{2}Nz_{2}T_{cm}
\end{equation}
\begin{equation}
k_{2}N^{2}w_{2}T_{cp} + 2k_{2}Nz_{2}T_{cm} = k_{3}N^{2}w_{3}T_{cp} + 2k_{3}Nz_{3}T_{cm}
\end{equation}
\begin{equation}
k_{i}N^{2}w_{i}T_{cp} + 2k_{i}Nz_{i}T_{cm} = k_{i+1}N^{2}w_{i+1}T_{cp} + 2k_{i+1}Nz_{i+1}T_{cm}
\end{equation}
\[
\begin{array}{c}
\vdots
\end{array}
\]
\begin{equation}
k_{p-1}N^{2}w_{p-1}T_{cp} + 2k_{p-1}Nz_{p-1}T_{cm} = k_{p}N^{2}w_{p}T_{cp} + 2k_{p}Nz_{p}T_{cm}
\end{equation}
\begin{equation}
k_{1}+k_{2}+k_{3}+\cdots+k_{p} = N
\end{equation}

\begin{figure}[H]
\centering
  \includegraphics[width=3.6in]{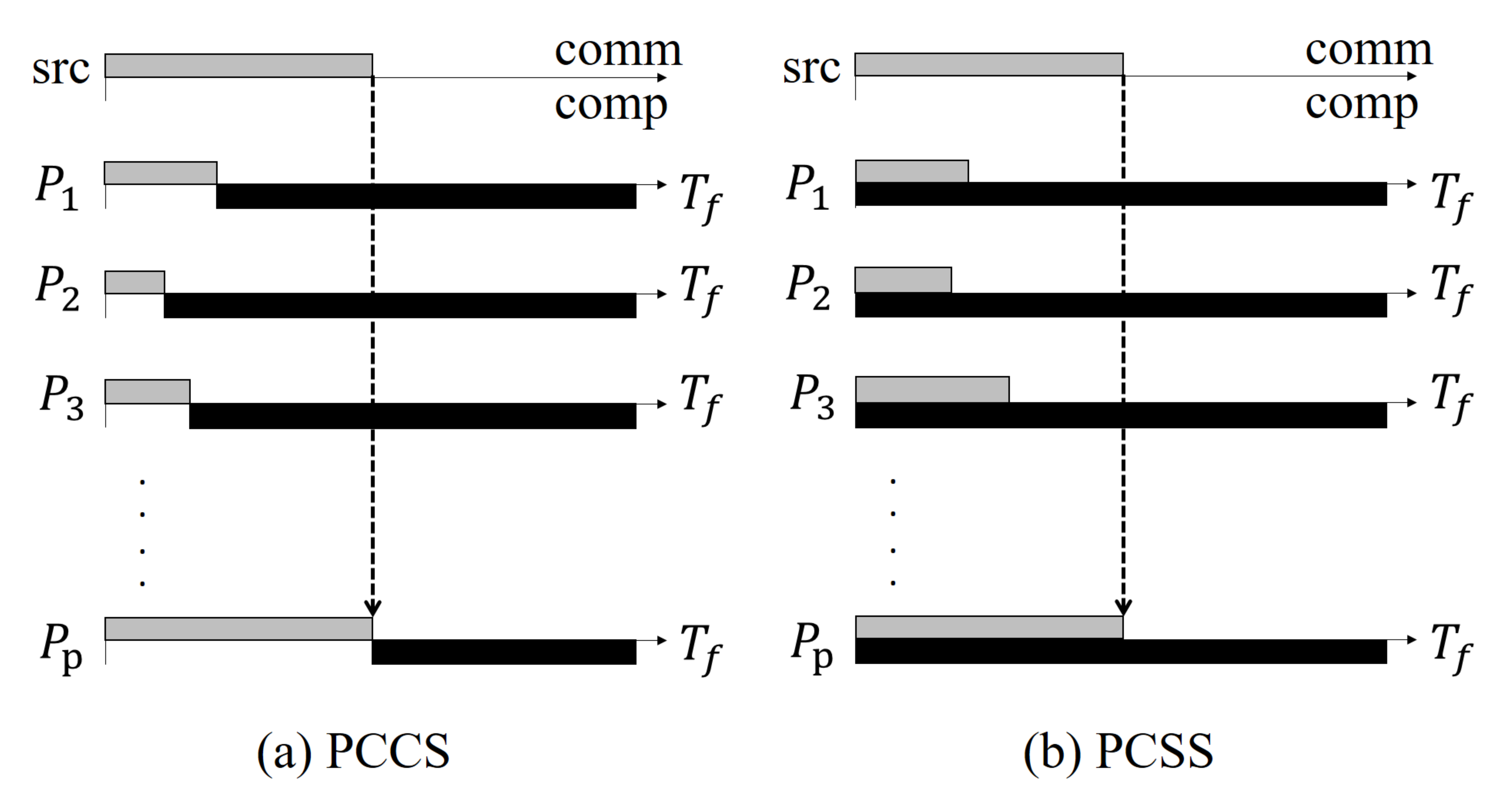}
  \caption{PCCS mode and PCSS mode.}\label{fig:PCCSandPCSS}
\end{figure}

Solving the equations, we have:
\begin{equation}
k_{i} = \prod_{j=2}^{i}\frac{Nw_{j-1}T_{cp} + 2z_{j-1}T_{cm}}{Nw_{j}T_{cp} + 2z_{j}T_{cm}}k_{1}, i=2,3,...p
\end{equation}

where $k_{1}$ is defined as:

\begin{equation}
k_{1} = \frac{N}{1+\sum_{i=2}^{p}\prod_{j=2}^{i}\frac{Nw_{j-1}T_{cp} + 2z_{j-1}T_{cm}}{Nw_{j}T_{cp} + 2z_{j}T_{cm}}}
\end{equation}

The overall finishing time of the whole network is:
\begin{equation}
T_{f} = k_{1}N^{2}w_{1}T_{cp} + 2k_{1}Nz_{1}T_{cm}
\end{equation}

\subsection{Parallel Communication Simultaneous Start}
The case where source can transmit data to all processors at the same time, and communication can overlap with computation is shown in Figure \ref{fig:PCCSandPCSS}(b). All processor start processing and end processing at the same time, and the size of load should be proportional to each processor's computing speed.
\begin{equation}
k_{i}N^{2}w_{i}T_{cp} = k_{i-1}N^{2}w_{i-1}T_{cp}, i=2,3..., p
\end{equation}
\begin{equation}
k_{1}+k_{2}+k_{3}+\cdots+k_{p} = N
\end{equation}
Solving the equations, we have:
\begin{equation}
k_{i} = \prod_{j=2}^{i} \frac{w_{j-1}}{w_{j}}k_{1}, i=2,3,...,p
\end{equation}

\begin{equation}
k_{1} = \frac{N}{1+\sum_{i=2}^{p}\prod_{j=2}^{i}\frac{w_{j-1}}{w_{j}}}
\end{equation}

\begin{equation}
T_{f} = k_{1}N^{2}w_{1}T_{cp}
\end{equation}

\subsection{Integer Adjustment}
By solving the above equations, we get a set of real numbers $\{k_{i}\}$ that generates the minimal task finishing time. The next step is to find the closest integer solution. We'll have a deeper discussion of how we obtain this closest integer solution from real number optimal solution, in the next section after we discuss \emph{multi-neighbor networks}. Because we'll face the same problem there, too. Here we'll provide a simple heuristic as shown below.

We can round off the real number solution first, such that a processor gets the whole row/column assignment if it takes more than half of the fractional part of that row/column in the real number optimal solution. After the rounding off process, if the sum of each $k_{i}$ fails to equal $N$, then we sort the processors in ascending order of their actual finish time $T_{f}(i)$. If the sum is less than $N$, then from the processor that has the smallest $T_{f}(i)$, we assign an extra row or column to that processor until the sum equals $N$. Otherwise, starting from the processor that has the largest $T_{f}(i)$, we remove one row/column from that processor until the sum equals $N$. After the process we obtain the actual integer assignment, we can further update the overall task finishing time.

\section{LBP In Multi-Neighbor Networks}
\label{chap:LBPMultiNeighbor}
In this subsection, we discuss the load balance strategy for another type of network, in which each node is allowed to receive loads from more than one of its neighbors. We call this type of network \emph{multi-neighbor networks}, in contrast to the \emph{single-neighbor networks} discussed previously. The target of the load balancing strategy is still minimizing the task finishing time. The definition of variables and constants are listed in the following Table \ref{table:variablesMultiNeighbor} and Table \ref{table:constantsMultiNeighbor}.

\begin{table}[H]
\centering
\caption{Table of the Variables for Chapter \ref{chap:LBPMultiNeighbor}}
\begin{tabular}{l|l}
\hline
Variables & Meaning\\
\hline
$T_{s}(i)$ & The start time of the $i$th node in the network\\
$T_{f}(i)$ & The finish time of the $i$th node in the network\\
$k_{i}$ & The number of columns or rows of the\\ & multiplier matrix that are assigned to $i$th node\\
$\phi(i,j)$ & The volume of load transmitted from\\ & node $i$ to node $j$ \\
\hline
\end{tabular}
\label{table:variablesMultiNeighbor}
\vspace{-0.15in}
\end{table}

\begin{table}[H]
\centering
\caption{Table of the Constants for Chapter \ref{chap:LBPMultiNeighbor}}
\begin{tabular}{l|l}
\hline
Constants & Meaning\\
\hline
$G(V,E)$ & The mesh network with fixed topology\\
$S$ & The set of source nodes\\
$z(i,j)$ & The inverse of the link speed of the link\\ & connecting node $i$ and $j$\\
$T_{cm}$ & Communication intensity constant \\
$w(i)$ & The inverse of the computing speed\\ & of $i$th processor \\
$T_{cp}$ & Computing intensity constant \\
$N$ & The size of both square multiplier matrices \\
$p$ & The number of nodes in this mesh network \\
$\tau(i,j)$ & specifies the position relationship\\ & of two nodes $i$ and $j$ in the network. \\
 & $\tau(i,j)=1$ if $i$ is in a position that should \\ & transmit to $j$, and $\tau(i,j)=0$ if otherwise \\
$D_{i}$ & specifies the storage size of node $i$. According \\ & to \ref{sec:memory}, this can be memory or hard disk \\
\hline
\end{tabular}
\label{table:constantsMultiNeighbor}
\vspace{-0.15in}
\end{table}

Note that in \emph{multi-neighbor networks}, it is not applicable to apply the equal processing time load balancing strategy such as the one addressed in Theorem \ref{SingleNeighborNetworkLoadBalancingTheorem}. Zhang et al. \cite{Zhangzeming} analyze a similar problem. To summarise, if node $i$ receives load from only one of its neighbors, says, $j$, we can obtain the following equation, $T_{s}(i) = T_{s}(j) + \phi(j,i)z(j,i)T_{cm}$, meaning that node $i$ starts processing when it finishes receiving load from node $j$. However, if node $i$ also receives load from another neighboring node, say $j'$, the above equation may fail to stand, because the two neighbor nodes $j$ and $j'$ might not finish transmitting at the same time, i.e,
\begin{equation}
T_{s}(j) + \phi(j,i)z(j,i)T_{cm} \neq T_{s}(j') + \phi(j',i)z(j',i)T_{cm}
\end{equation}
and we can not determine node $i$'s start time.

Theorem \ref{SingleNeighborNetworkLoadBalancingTheorem} can optimize task finishing time for \emph{single-neighbor networks} like tree, multi-level tree, etc. For \emph{multi-neighbor networks} like multi-root tree, mesh, ring, torus, hypercube, since we can not apply Theorem \ref{SingleNeighborNetworkLoadBalancingTheorem}, we formulate the load balancing problem as an optimization problem, called \emph{Minimize Maximum Finishing Time in Layer Based Partition(MMFT-LBP)} problem.

A mesh is a typical \emph{multi-neighbor network} where one single node can receive data from multiple neighbors. In the following part we discuss solving MMFT-LBP problem in mesh networks, to shed light on solving MMFT-LBP in other \emph{multi-neighbor networks}.

\subsection{MMFT-LBP In Mesh}

\begin{figure}[H]
\centering
  \includegraphics[width=1.2in]{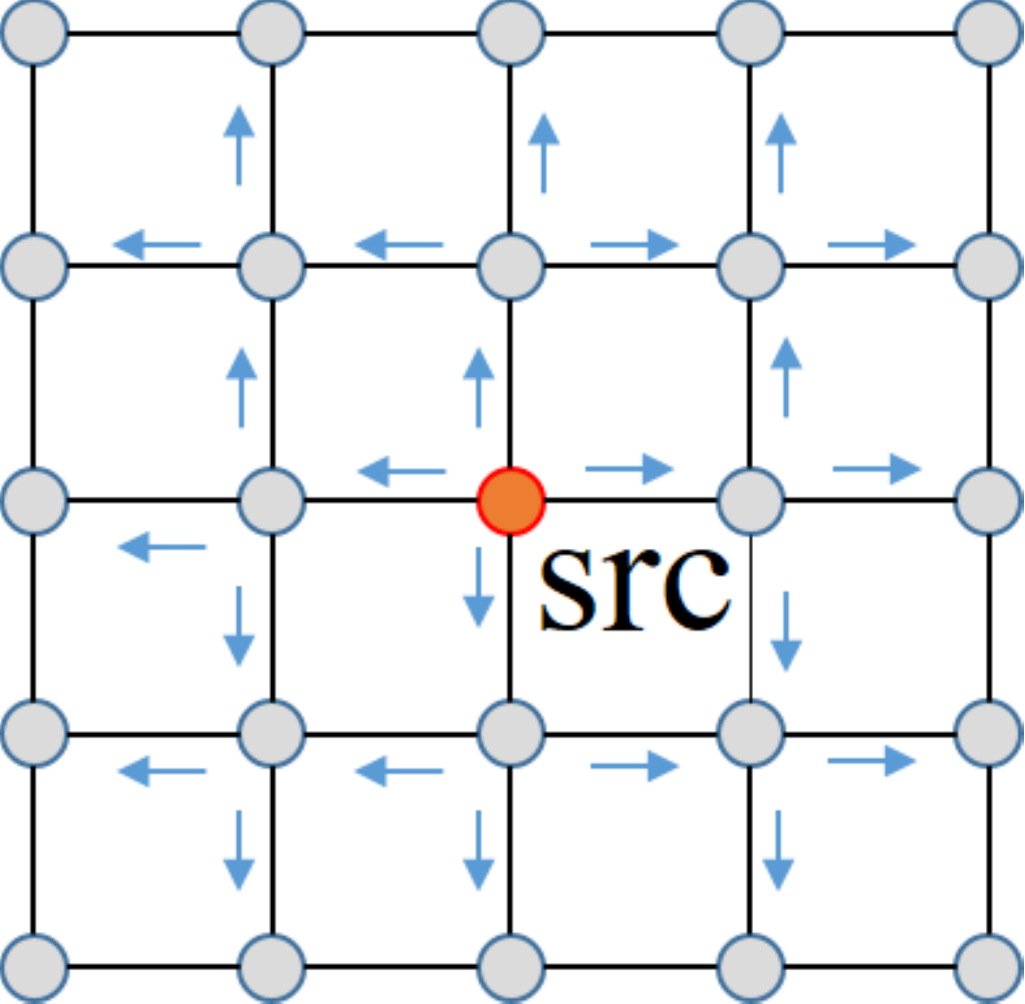}
  \caption{data flow in mesh network.}\label{fig:meshFlow}
\end{figure}

In this subsection we focus on the \emph{MMFT-LBP} problem in mesh networks. We denote the mesh network as a graph $G(V,E)$, shown in figure \ref{fig:meshFlow}. The graph contains $p$ nodes which forms a dimension of $X * Y$, where $p = X * Y$. Generally speaking, the mesh has better performance in scheduling if the source is closer to geometric center. So we assume that the source is located at the center of this mesh network, and divides the mesh into four quadrants. The specific data flow pattern in each quadrant is shown in Figure \ref{fig:meshFlow}.

For the limitation of scope of this paper, we only analyze PCCS mode, that is, data is forwarded in the network in the `parallel communication and consecutive start' pattern. `Parallel communication' allows each node to talk to its multiple neighbors at the same time. `Consecutive start' means that each node has to wait until it receives its whole share of data before it can start computing. In short, we utilize PCCS mode as each node's communication and processing model in our \emph{multi-neighbor networks}, and leave the other communication patterns for future study.

We present the \emph{MMFT-LBP} problem in the following.

\noindent
\textbf{MMFT-LBP}

\noindent
\textbf{Variables:} $\{k_{i}\}$, $\{T_{s}(i)\}$, $\{T_{f}(i)\}$, $\{\phi(i,j)\}$, $T_{f}$

\noindent
\textbf{Objective:}

\begin{equation}
Minimize: \max_{i \in G(V,E)} \{T_{f}(i)\}
\label{obj}
\end{equation}
\textbf{Constraint:}
\begin{equation}
T_{s}(i) = 0, ~\forall i \in S
\label{src start}
\end{equation}
\begin{equation}
T_{s}(i) = \max_{j \in G(V,E)} \{\tau(j,i)\big{[}T_{s}(j) + \phi(j,i)\cdot z(j,i)T_{cm}\big{]}\}, ~\forall i \notin S
\label{nonsrc start}
\end{equation}
\begin{equation}
T_{f}(i) = T_{s}(i) + k_{i}N^{2}w(i)T_{cp}, ~\forall i \in G(V,E)
\label{finish time}
\end{equation}
\begin{equation}
2N^{2} - \sum_{j \in G(V,E)} \phi(i, j) = 0, ~\forall i\in S
\label{src part}
\end{equation}
\begin{equation}
\sum_{j \in G(V,E)} \phi(j,i) - \sum_{j' \in G(V,E)} \phi(i,j') = 2k_{i}N, ~\forall i\notin S
\label{nonsrc part}
\end{equation}
\begin{equation}
\phi(i,j) \geq 0, ~\forall ~\tau(i,j) = 1
\label{phi geq 0}
\end{equation}
\begin{equation}
\phi(i,j) = 0, ~\forall ~\tau(i,j) = 0
\label{phi eq 0}
\end{equation}
\begin{equation}
k_{i} \in \mathbb{Z_+}, ~\forall i\in G(V,E)
\label{k positive integer}
\end{equation}
\begin{equation}
k_{i} = 0, ~\forall i\in S
\label{src share}
\end{equation}
\begin{equation}
2k_{i}N + N^{2} \leq D_{i}, ~\forall i \notin S
\label{storage condition}
\end{equation}
\begin{equation}
\sum_{i=1}^{p} k_{i} = N
\label{normalization}
\end{equation}

\noindent
\textbf{Remarks:}
\begin{itemize}
\item The objective of \emph{MMFT-LBP}(\ref{obj}) is to minimize the maximum finishing time of the mesh network.
\item Constraint (\ref{src start}) specifies the start time of the source. In this paper's case, there's only one node in set S.
\item Constraint (\ref{nonsrc start}) specifies that the start time of those non-source nodes should be the maximum time that they finish receiving all loads from their adjacent neighbors.
\item $\tau(i,j)$ is a constant once the mesh network is determined and fixed.
\item Constraint (\ref{finish time}) defines each node's finishing time. According to LBP, each node's computing load is the total number of multiplication it conducts $k_{i}*N^{2}$.
\item Constraint (\ref{src part}) shows that because the source does not take part in processing, it sends out all of its load: the two multiplier matrices. And each entry of the matrices is sent only once.
\item Constraint (\ref{nonsrc part}) defines the amount of load taken by those non-source nodes.
\item Constraint (\ref{phi geq 0}) - (\ref{phi eq 0}) list different values of $\phi(i,j)$ in different cases. (\ref{phi geq 0}) is true if node $i$ is in a position that should send load to node $j$. (\ref{phi eq 0}) applies when two nodes are not adjacent or node $j$ is in a position that should send load to node $i$.
\item The $k_{i}$ in constraint (\ref{k positive integer}) is the number of columns taken by each node, and should be integers.
\item Constraint (\ref{src share}) shows that the source node does not take part in processing so its $k_{i} = 0$, and $T_{f}(i) = 0$ through constraint (\ref{finish time}).
\item Constraint (\ref{storage condition}) shows that for each node, its storage size should at least be able to store the load from two multiplier matrices ($2k_{i}N$), and the result matrix ($N^2$).
\item Constraint (\ref{normalization}) is the normalization constraint. The number of columns taken by each node should sum up to be the side length of the multiplier matrix.
\end{itemize}

Both the objective function and constraints contain maximum form of formulas. To solve it, we firstly reorganize the problem and remove those maximum forms of formulas, and show that both the original and reorganized form have the same optimal solution.

\subsection{MFT-LBP In Mesh}
The objective function of \emph{MMFT-LBP} problem contains maximum form of formulas, which makes it hard to solve directly. To obtain the optimal solution of \emph{MMFT-LBP} problem, the first step is to reorganize the objective function.

We introduce one additional unknown variable $T_{f}$, and a set of constraints $T_{f} \geq T_{f}(i), \forall i \in G(V,E)$, which ensure $T_{f}$ to be no earlier than any node's finishing time $T_{f}(i)$. Then we transform the objective function from (\ref{obj}) to $Minimize: T_{f}$. The optimal solution to the original problem is still the optimal solution to the transformed problem, because when $T_{f}$ is minimized, $T_{f} = \max_{i \in G(V,E)} \{T_{f}(i)\}$.

Similarly, we relax constraint (\ref{nonsrc start}) to be the following linear inequality.
\begin{equation}
T_{s}(i) \geq \tau(j,i)[T_{s}(j) + \phi(j,i)\cdot z(j,i)T_{cm}], ~\forall i \notin S
\label{nonsrc start relax}
\end{equation}

This inequality (\ref{nonsrc start relax}) implies processor $i$ does not necessarily need to begin processing immediately when it finishes receiving all the data, it can choose to hold the data for a while and then start processing. Therefore, constraint (\ref{nonsrc start relax}) increases the solution space defined by constraint (\ref{nonsrc start}) to incorporate the following:

\begin{equation}
T_{s}(i) > \max_{j \in G(V,E)} \{\tau(j,i)\big{[}T_{s}(j) + \phi(j,i)\cdot z(j,i)T_{cm}\big{]}\}, ~\forall i \notin S
\label{nonsrc start relax2}
\end{equation}

Even if inequality (\ref{nonsrc start relax}) increases the feasible solution space, the optimal solution still remains the same. This is because the target of this problem is to minimize the overall finishing time, the minimal finishing time is achieved always when each node starts processing and forwarding once it completes receiving load from its neighbors, even if it is allowed to hold on for a while before it starts processing and forwarding.

In \emph{MMFT-LBP} problem statements, if we use constraint (\ref{nonsrc start relax}) to replace constraint (\ref{nonsrc start}), we get a new problem called \emph{Minimize Finish Time with Layer Based Partition(MFT-LBP)}. As we have discussed, we have the following theorem.

\begin{theorem}
\emph{(Maximization Relaxation Theorem)}
\label{MaximizationRelaxationTheorem}
\emph{MMFT-LBP} problem has the same optimal solution with \emph{MFT-LBP} problem.
\end{theorem}

\noindent
The full problem statement is presented as follows.

\noindent
\textbf{MFT-LBP}

\noindent
\textbf{Variables:} $\{k_{i}\}$, $\{T_{s}(i)\}$, $\{T_{f}(i)\}$, $\{\phi(i,j)\}$, $T_{f}$

\noindent
\textbf{Objective:}

\begin{equation}
Minimize: ~T_{f}
\label{MFT-obj}
\end{equation}
\textbf{Constraint:}
\begin{equation}
T_{s}(i) = 0, ~\forall i \in S
\label{MFT-src start}
\end{equation}
\begin{equation}
T_{s}(i) \geq \tau(j,i)\big{[}T_{s}(j) + \phi(j,i)\cdot z(j,i)T_{cm}\big{]}, ~\forall i \notin S
\label{MFT-nonsrc start}
\end{equation}
\begin{equation}
T_{f}(i) = T_{s}(i) + k_{i}N^{2}w(i)T_{cp}, ~\forall i \in G(V,E)
\label{MFT-finish time}
\end{equation}
\begin{equation}
2N^{2} - \sum_{j \in G(V,E)} \phi(i, j) = 0, ~\forall i\in S
\label{MFT-src part}
\end{equation}
\begin{equation}
\sum_{j \in G(V,E)} \phi(j,i) - \sum_{j' \in G(V,E)} \phi(i,j') = 2k_{i}N, ~\forall i\notin S
\label{MFT-nonsrc part}
\end{equation}
\begin{equation}
\phi(i,j) \geq 0, ~\forall ~\tau(i,j) = 1
\label{MFT-phi geq 0}
\end{equation}
\begin{equation}
\phi(i,j) = 0, ~\forall ~\tau(i,j) = 0
\label{MFT-phi eq 0}
\end{equation}
\begin{equation}
k_{i} \in \mathbb Z_+, ~\forall i\in G(V,E)
\label{MFT-k geq 0}
\end{equation}
\begin{equation}
k_{i}=0, \forall i \in S
\label{MFT-src share}
\end{equation}
\begin{equation}
2k_{i}N + N^{2} \leq D_{i}, ~\forall i \notin S
\label{MFT-storage condition}
\end{equation}
\begin{equation}
\sum_{i=1}^{p} k_{i} = N
\label{MFT-normalization}
\end{equation}
\begin{equation}
T_{f} \geq T_{f}(i), \forall i \in G(V,E)
\label{MFT-tf}
\end{equation}

\noindent
\textbf{Remarks:}
As constraint (\ref{MFT-k geq 0}) shows, $\{k_{i}\}$ are positive integers, which make \emph{MFT-LBP} problem a \emph{Mixed Integer Non-linear Programming} problem. To solve it, we propose an algorithm called \emph{Phased Minimization of Finish Time with Layer Based Partition}(PMFT-LBP).

\subsection{PMFT-LBP}
\noindent
\textbf{PMFT-LBP:}
The \emph{PMFT-LBP} algorithm contains the following three phases. In Phase I, it relaxes integers $\{k_{i}\}$ to real numbers and solves the relaxed linear programming problem. In Phase II, based on the optimal real number solution obtained in Phase I, \emph{PMFT-LBP} determines a feasible integer solution for the original \emph{PMFT-LBP} problem, which is close to the optimal real number solution. In Phase III, starting from the feasible integer solution obtained in Phase II, \emph{PMFT-LBP} conducts a ``neighbor search" and seeks for local optimal feasible integer solution.

\noindent
\textbf{Phase I.} In this phase, the \emph{PMFT-LBP} algorithm relaxes condition (\ref{MFT-k geq 0}) to be a positive real number
\begin{equation}
k_{i} \geq 0, ~\forall i \in G(V,E)
\end{equation}
and solves a relaxed version of the \emph{MFT-LBP} problem, called \emph{MFT-LBP-relax}.

With this relaxation, all constraints are either linear equality or linear inequality, forming a convex polygon feasible region. Its objective function is also linear defined on this polyhedron. Hence, \emph{MFT-LBP} is a LP problem and can be solved to get the optimal real number solution $Q^{*} = \{\{k_{i}\}$, $\{T_{s}(i)\}$, $\{T_{f}(i)\}$, $\{\phi(i,j)\}$, $T_{f}\}$.

\noindent
\textbf{Phase II.} In this phase, \emph{PMFT-LBP} calls an algorithm called \emph{finds an integer feasible solution(FIFS)} based on the optimal real number solution $Q^{*}$ obtained from phase I. we firstly round off each $k_{i}$ to its closest integer. The intuition here is that if the real number optimal solution assigns the larger portion of a column to one processor, then assigning that processor with the entire column will have a higher chance to get ``closer" to optimality. It is vice versa that if the processor is assigned with the smaller portion of a column, we shall remove its share of this column.

\begin{algorithm}[!t]
\caption{PMFT-LBP}
\label{PMFT-LBP}
\begin{algorithmic}[1]
\Function{PMFT-LBP()}{}
\BState \emph{Phase I:}
\State solve \emph{MFT-LBP-relax}, get optimal real number solution $\{\{k_{i}\}$, $\{T_{f}(i)\}$, $\{T_{s}(i)\}$, $\{\phi(i,j)\}$, $T_{f}\}$.

\BState \emph{Phase II:}
\State call \emph{FIFS} to get $\{k_{i}'\}$, a feasible integer schedule.

\BState \emph{Phase III:}
\State use current schedule as start point: $p_{cur} \Leftarrow \{k_{i}'\}$
\State re-solve LP for $p_{cur}$.
\While{true}
\State choose node $a$ with $T_{f}'(a) = \max\{T_{f}'(i)\}$
\State $k_{a}'' \leftarrow k_{a}'-1$
\State choose node $b$ with $T_{f}'(b) = \min\{T_{f}'(i)\}$
\State $k_{b}'' \leftarrow k_{b}'+1$
\State construct $\{k_{i}''\}$ using $k_{a}'', k_{b}''$ replacing $k_{a}', k_{b}'$:
\State $\{k_{i}''\} \leftarrow \{k_{1}', k_{2}', ... k_{a}'', ..., k_{b}'', ...k_{p}'\}$
\State get neighbor: $p_{nb} \Leftarrow \{k_{i}''\}$
\State re-solve LP for $p_{nb}$.
\If {$T_{f}'\langle p_{cur}\rangle  < T_{f}'\langle p_{nb}\rangle$} break
\Else ~{$p_{cur} \leftarrow p_{nb}$}
\EndIf
\EndWhile
\State
\Return optimal schedule $p_{cur}$.
\EndFunction
\end{algorithmic}
\end{algorithm}

However, the rounding off process alone may not guarantee a feasible integer solution, because it may result in the sum of all $k_{i}$ fail to equal the multiplier matrix's side length $N$, constraint (\ref{MFT-normalization}). To resolve this problem, we conduct a subtle adjustment. For cases in which the sum is greater than $N$, meaning that there exists duplicate assignments, we shall reduce work load from some processors. Since the processor with the longest finishing time is the bottleneck affecting the overall finishing time, it has the highest priority to reduce its share of loads. On the contrary, for cases in which the sum is less than $N$, meaning that some rows/columns haven't been assigned to any processor, processor currently with the shortest running time has the highest priority to take up the responsibility.

We don't make the adjustment to one processor all at once. Instead, we conduct the adjustment iteratively. Every iteration we only adjust one row/column, then we update each processor's $\{T_{s}(i)\}$, $\{T_{f}(i)\}$, $\{\phi(i,j)\}$ to determine the processor to conduct adjustment on for the next round of iteration. The processor with the longest processing time currently will be the one to remove a row/column from in the next round of iteration, and the processor with the shortest processing time will be the one to take the extra load in the next iteration.

When the sum of all $k_{i}$ equals $N$, we finally find an integer feasible solution $\{k_{i}'\}$, $\{T_{s}'(i)\}$, $\{T_{f}'(i)\}$, $\{\phi'(i,j)\}$.

\begin{algorithm}[!t]
\caption{FIFS Algorithm}
\label{FIFS Algorithm}
\begin{algorithmic}[1]
\Function{FIFS ($\{k_{i}\}$, $\{T_{s}(i)\}$, $\{T_{f}(i)\}$, $\{\phi(i,j)\}$, $T_{f}$)}{}
\For {each $k_{i} \in \{k_{i}\}$}
\State $k_{i}' \leftarrow round(k_{i})$
\EndFor
\State $Sum \leftarrow \sum_{i}\{k_{i}'\}$
\While{$Sum \neq N$}
\State use $\{k_{i}'\}$ as known variables, re-solve \emph{MFT-LBP} problem, get updated $\{T_{s}'(i)\}$, $\{T_{f}'(i)\}$, $\{\phi'(i,j)\}$.
\If {$Sum > N$}
\State choose node $j$ with $T_{f}'(j) = \max\{T_{f}'(i)\}$
\State $k_{j}' \leftarrow k_{j}'-1$
\State $Sum \leftarrow Sum-1$
\EndIf
\If {$Sum < N$}
\State choose node $j$ with $T_{f}'(j) = \min\{T_{f}'(i)\}$
\State $k_{j}' \leftarrow k_{j}'+1$
\State $Sum \leftarrow Sum+1$
\EndIf
\EndWhile
\State
\Return integer schedule $\{k_{i}'\}$.
\EndFunction
\end{algorithmic}
\end{algorithm}

\noindent
\textbf{Phase III.} In this phase, \emph{PMFT-LBP} conducts a so-called ``Neighbor Search" process to optimize the feasible solution obtained in phase II. The reason why this feasible solution still need optimization is that the ``Rounding Off" procedure in phase II might bring about bias and take our feasible solution ``away" from optimal. Therefore, a further optimization is necessary.

The ``Neighbor Search" runs in iterations. For the first iteration, it starts from the feasible integer schedule $\{k_{i}'\} = \{k_{1}', k_{2}', ... k_{a}', ..., k_{b}', ...k_{p}'\}$ obtained in phase II, and compares the current overall finishing time with its neighbors'. $\{k_{i}''\} = \{k_{1}'', k_{2}'', ... k_{a}'', ..., k_{b}'', ...k_{p}''\}$ is defined as one of $\{k_{i}'\}$'s neighbors if all $k_{i}'' = k_{i}'$ except for only two dimensions $k_{a}''$ and $k_{b}''$, where $k_{a}'' = k_{a}' - 1$ and $k_{b}'' = k_{b}' + 1$, such that $\sum_{i=1}^{p}k_{i}'' = \sum_{i=1}^{p}k_{i}' = N$. If one neighbor $\{k_{i}''\}$ offers shorter overall finishing time than $\{k_{i}\}'$ and the rest of its neighbors, we change load schedule from $\{k_{i}'\}$ to $\{k_{i}''\}$. Then, ``Neighbor Search" will use $\{k_{i}''\}$ as a new start point and begins a new iteration. If at some stage, no further optimization can be made towards minimizing overall finishing time, which means the schedule at that stage $\{k_{i}^{*}\}$ is a local minimal point that has the shortest overall finishing time among all of its neighbors, iteration stops and we say we find our optimal schedule.

\begin{algorithm}[!t]
\caption{MFT-LBP-heuristic}
\label{MFT-LBP-heuristic}
\begin{algorithmic}[1]
\Function{MFT-LBP-heuristic ()}{}
\State solve \emph{MFT-LBP-relax}, get optimal real number solution $\{k_{i}\}$, $\{T_{f}(i)\}$.
\State $\{k_{i}'\} \leftarrow round(\{k_{i}\})$
\State use $\{k_{i}'\}$ as known variables, re-solve \emph{MFT-LBP} problem, get updated $\{T_{f}'(i)\}$.
\State $Sum \leftarrow \sum_{i}\{k_{i}'\}$
\State $diff \leftarrow Sum-N$
\If {$diff < 0$}
\State sort $\{T_{f}'(i)\}$ in ascending order
\State corresponding $k_{i} \leftarrow k_{i}+1$ till $diff = 0$.
\Else ~sort $\{T_{f}'(i)\}$ in descending order
\State corresponding $k_{i} \leftarrow k_{i}-1$ till $diff = 0$.
\EndIf
\EndFunction
\end{algorithmic}
\end{algorithm}

\subsection{MFT-LBP-heuristic}

Considering the high time complexity of \emph{PMFT-LBP} algorithm because so many LP-based updates are involved, we propose a heuristic called \emph{MFT-LBP-heuristic}. This heuristic calculates a good feasible integer solution that is close to the optimal solution with a time complexity which is substantially reduced from \emph{PMFT-LBP}. The basic idea is that whether one single row/column should be assigned to one processor or another simply doesn't bring about much improvement of the result. However, simplifying these steps saves time complexity substantially.

Based on this idea, \emph{MFT-LBP-heuristic} only keeps phase I the same with \emph{PMFT-LBP} algorithm.
In phase II, after obtaining the optimal real-number schedule $\{k_{i}\}$, the heuristic rounds off each $k_{i}$ to get $\{k_{i}'\}$ slightly differently. The difference is, if the sum of $\{k_{i}'\}$ doesn't equal $N$, the heuristic sorts all processors in ascending order of their finishing time $\{T_{f}(i)'\}$, which forms an array $Arr[p]$. If the sum of $\{k_{i}'\}$ is less than $N$, then from the first element in $Arr[p]$, the heuristic keeps adding $1$ to each processor's $k_{i}'$ until $\sum_{i}\{k_{i}'\} = N$. Otherwise if the sum of $\{k_{i}'\}$ is greater than $N$, the heuristic starts from the last element of $Arr[p]$ and minus $1$ from each processor's $k_{i}'$ one by one until $\sum_{i}\{k_{i}'\} = N$. The adjusting process continues in a circular manner so if it reaches one end of $Arr[p]$, it jumps to the the other end of the array and continue the process again.

In phase III, \emph{PMFT-LBP} searches each neighbor of current schedule $\{k_{i}'\}$. The time complexity of searching process is $O(p^2)$, where $p$ is the number of processors.
Once $p$ is big, the scalability of the algorithm is poor. In order to speed up the local search process, we apply the concept of `gradient descent' here. At each iteration, we only look at the neighbor $\{k_{i}''\}$ that has the highest chance to decrease the overall finishing time from current schedule $\{k_{i}'\}$. We know that $\{k_{i}''\}$  differentiate from $\{k_{i}'\}$ by just two dimensions $k_{a}'' = k_{a}' - 1$ and $k_{b}'' = k_{b}' + 1$. If $k_{a}'$ is the schedule of the processor that currently takes the longest processing time, while $k_{b}'$ is the schedule from the processor that takes the shortest, $\{k_{i}''\}$ then stands the highest the chance to be the neighbor that has shortest overall finishing time. We compare $\{k_{i}''\}$'s $T_{f}$ with $\{k_{i}'\}$'s. If $\{k_{i}''\}$'s is shorter, we use $\{k_{i}''\}$ in the next iteration. Otherwise, since even $\{k_{i}''\}$ cannot further decrease the overall finishing time, the other neighbors probably cannot either. Therefore we take $\{k_{i}'\}$ as the optimal schedule we look for.

This heuristic only solves LP problems twice and reduces time complexity of the iterative LP-based update process in \emph{PMFT-LBP}. The result of this may sacrifice the overall finishing time of the mesh network a little bit, but reduce the algorithm's time complexity substantially by reducing a lot of time-consuming LP iterative updating processes.

As will be seen in next section, the performance of our heuristic is extremely close, and even in cases equal to \emph{PMFT-LBP} both in terms of communication volume and finish time.

\section{Performance Evaluation}\label{performance evaluation}
\subsection{Performance Evaluation of Star Network}

In this subsection, we study the performance of our \emph{layer based partition} scheme in a star network, which is an instance of \emph{single neighbor network.} In each iteration, we randomly generate two $N*N$ matrices, and a heterogeneous star network of processors to conduct the matrix multiplication.

\noindent
\textbf{Star Network.} The star network contains one source and multiple children connected to the source. As mentioned previously, we assume the source node only transmits load, but does not take part in computing. In our simulation, we use a star network containing 16 children, where each child's unit processing time $wTcp$ is uniformly distributed in the range of (0.0005, 0.0008), and each link's unit transmission time $zTcm$ is uniformly distributed in the range of (0.0002, 0.0005).

\noindent
\textbf{Communication and Processing mode.} In order to better compare our algorithm with some other related algorithms, we use PCCS mode as the communication and processing mode here, meaning the source can communicate with each processor in a parallel manner whereas communication cannot overlap with computation.

\noindent
\textbf{Matrices.} The side length of the matrices in our simulation goes from 100 to 1000. When analyzing the performance of each algorithm over matrix size, each data point is an average of 10 independent experiments over 10 independently different star network.

\subsubsection{Evaluation Metrics}

\noindent
\textbf{Total Communication Volume.} Total communication volume is defined as the sum of data volume transmitted on each link.

\noindent
\textbf{Task Finishing Time.} The task finishing time in star network is defined as the time period from the source starts to send data till the last processor finishes working.

\subsubsection{Comparison Algorithms}
We compare our \emph{layer based partition} algorithm to a couple of typical \emph{rectangular partition} algorithms.

\noindent
\textbf{Even-Col.} Even-Col is a naive \emph{rectangular partition} algorithm that simply partitions the matrix into equivalent columns.

\noindent
\textbf{PERI-SUM.} Beaumont {\em et al.} \cite{Beaumont2} deal with the geometric problem of partition the unit square into $p$ rectangles of given area $s_{1}$, $s_{2}$, $s_{3}$, etc, so as to minimize the the sum of the perimeters(PERI-SUM) of the $p$ rectangles, which is proportional to communication volume. Beaumont {\em et al.} further propose the communication lower bound, and introduce a $1.75$-approximation algorithm. We use it in our comparison.

\noindent
\textbf{Recursive.} Nagamochi {\em et al.} \cite{Nagamochi} introduce a recursive partitioning technique on the basis of \textbf{PERI-SUM}, and improve the approximation ratio from $1.75$ to $1.25$.

\noindent
\textbf{NRRP.} Beaumont {\em et al.} \cite{Beaumont3} combine the idea of non-\emph{rectangular partition}ing from DeFlumere\cite{DeFlumere2} and recursive partitioning algorithm proposed by Nagamochi\cite{Nagamochi}. The combination of these two ingredients lead to an improvement of the approximation ratio down to $\frac{2}{\sqrt{3}} \simeq 1.15$.

\noindent
\textbf{Lower Bound.} Ballard {\em et al.} \cite{Ballard} proposed the lower bound of communication volume in \emph{rectangular partition} to be $2\sum_{i=1}^{p}(\sqrt{s_{i}})$. We use this lower bound to compare with our LBP's communication volume.

\subsubsection{Evaluation Result on Star Network}
\textbf{Communication Volume with Increasing Matrix Size.} As one of our most significant contributions, the simulation displays overwhelming superiority of our \emph{layer based partition} algorithm over \emph{rectangular partition} algorithms in total communication volume. Figure \ref{fig:StarNetwork}(a) compares the total communication volume of each algorithm. While the total communication volume generally increases along with the expansion of matrix size $N$, LBP generates the smallest total communication volume. Specifically, when the matrix size reaches $1000$, the total communication volume of \emph{layer based partition} reduces $75\%$ from the lower bound of the \emph{rectangular partition}. Meanwhile, as stated in \cite{Beaumont3}, the lower bound of \emph{rectangular partition} is too optimistic to reach in actual practice, and the real best \emph{rectangular partition} result is obtained actually in NPPR\cite{Beaumont3}, Recursive\cite{Nagamochi} and PERI-SUM\cite{Beaumont2}. Genearally, LBP presents a total communication volume that reduces $78\%$ from NPPR, $79.7\%$ from Recursive, and $85.1\%$ from PERI-SUM, respectively. We observe that LBP keeps this ratio over \emph{rectangular partition} algorithms along with the increase of matrix side length. Moreover, when the star network gets larger, this ratio gain becomes even bigger. We attribute these to the advantage of \emph{layer based partition} over \emph{rectangular partition} in communication volume.

\noindent
\textbf{Finishing Time with Increasing Matrix Size.} Figure \ref{fig:StarNetwork}(b) shows the overall finishing time of each algorithm. We observe that while all algorithm's finishing time increase with matrix size, LBP, PERI-SUM, Recursive, NPPR present similar curves, with their overall finishing time much smaller than that of Even-Col. Specifically, when the matrix size is $1000$, the overall finishing time of those four algorithm is about $40\%$ smaller than that of Even-Col. One reason for this is that those four algorithms, LBP, PERI-SUM, Recursive, NPPR all achieve load balance when scheduling the load. No matter dividing the matrix into layers or rectangles, each of the four algorithms makes sure that each share of load is proportional to that processor's computing ability.

\begin{figure}[ht]
\centering
    \subfigure[Communication Volume]{
        \includegraphics[width=2.1in]{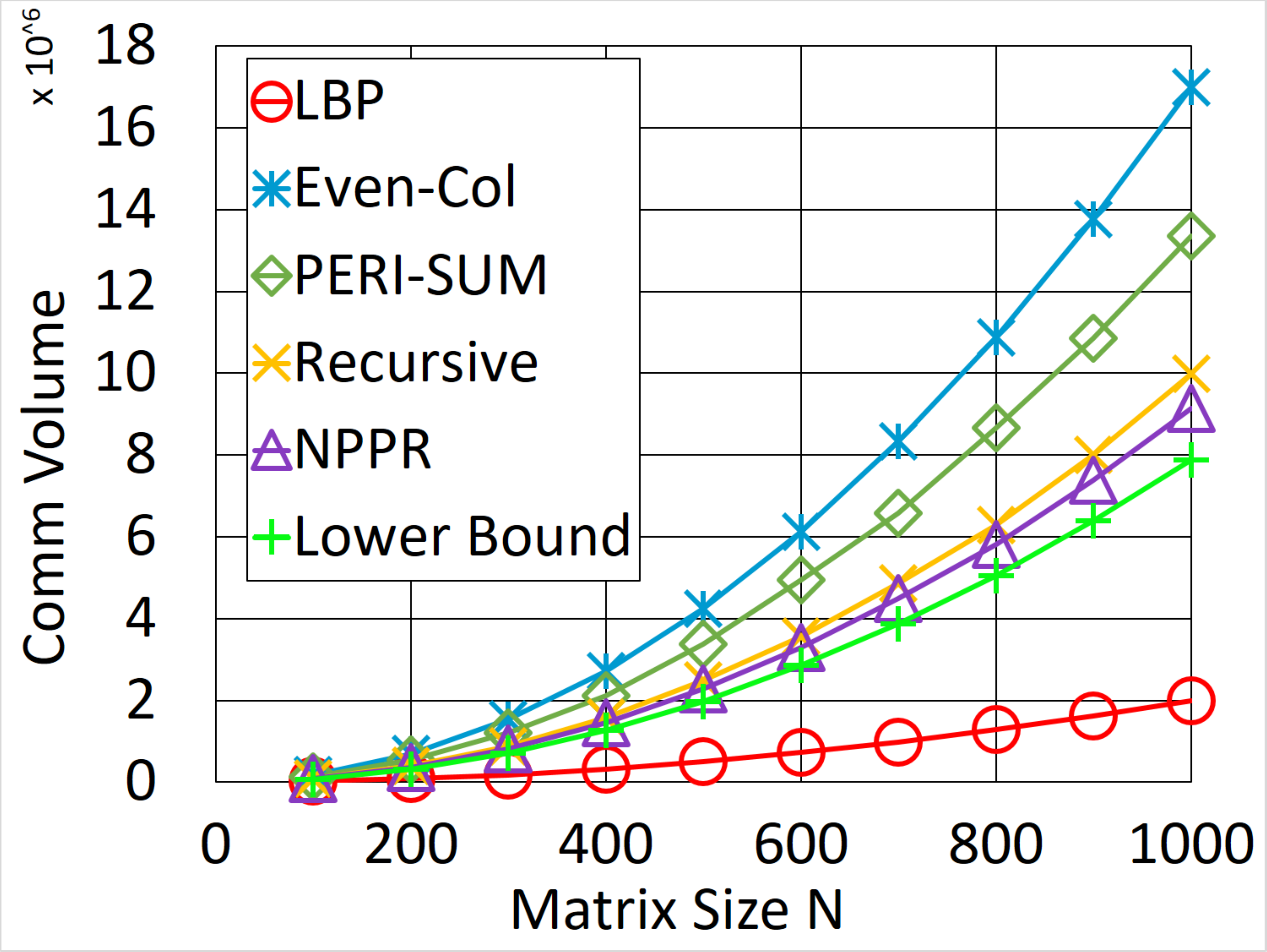}
        \label{fig:StarComvol}
    }
    \subfigure[Finishing Time]{
        \includegraphics[width=2.1in]{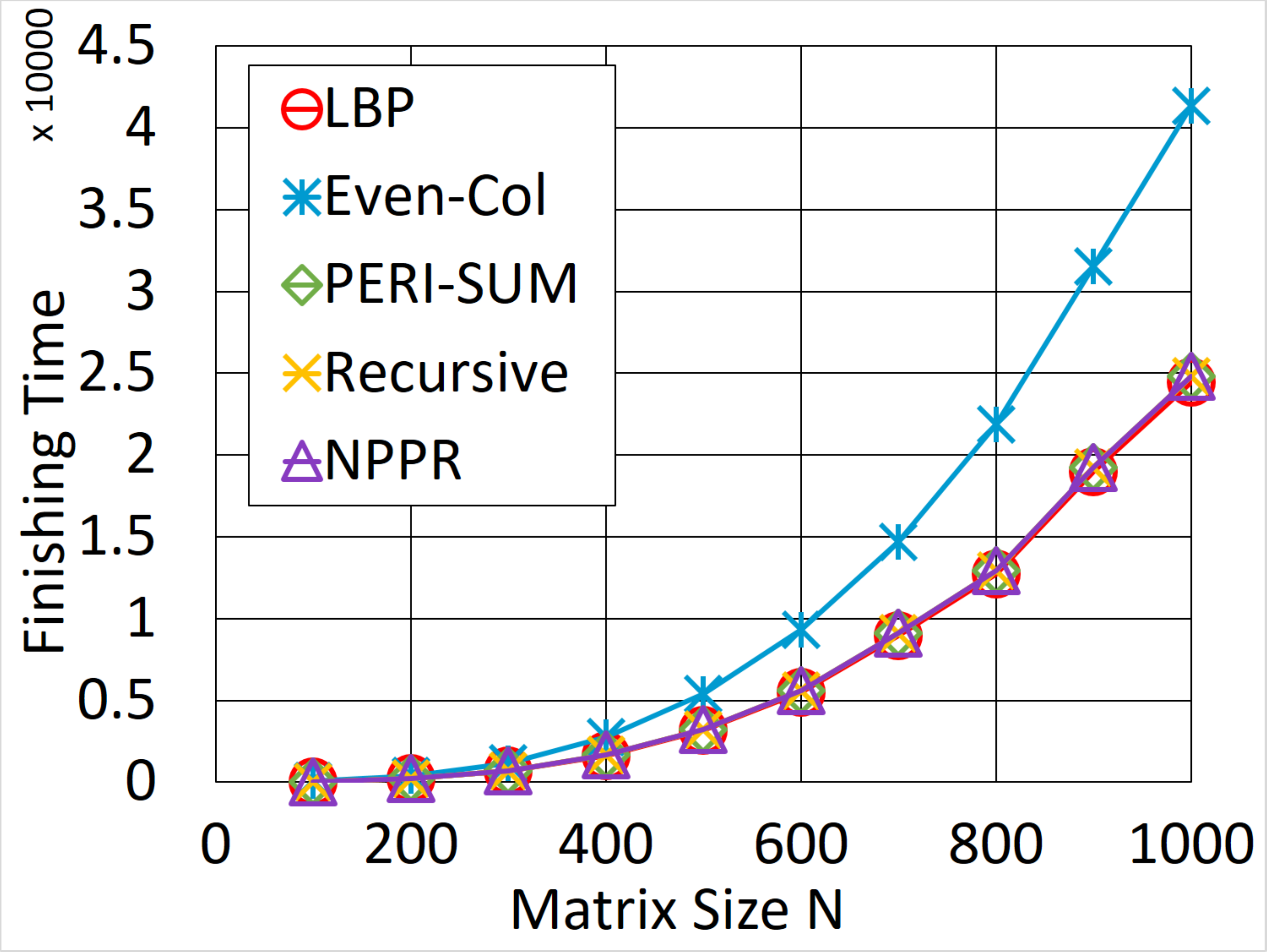}
        \label{fig:StarFinishingTime}
    }
\caption{Performance comparison on 16-node star network}
\label{fig:StarNetwork}
\end{figure}

\noindent
\textbf{Summary.} In this subsection we compare our \emph{layer based partition} algorithm with \emph{rectangular partition} algorithms in terms of total communication volume and task finishing time. The simulation results proves that our \emph{layer based partition} algorithm generates a total communication volume that is substantially reduced from the state-of-the-art \emph{rectangular partition} algorithms. In the meanwhile, LBP also reaches load balance and generates a total overall finishing time that is as low as the other algorithms.

\subsection{Performance Evaluation of Mesh}

\begin{figure*}
\begin{center}
  \subfigure[5*5 Mesh]{
  \includegraphics[width=2.1in]{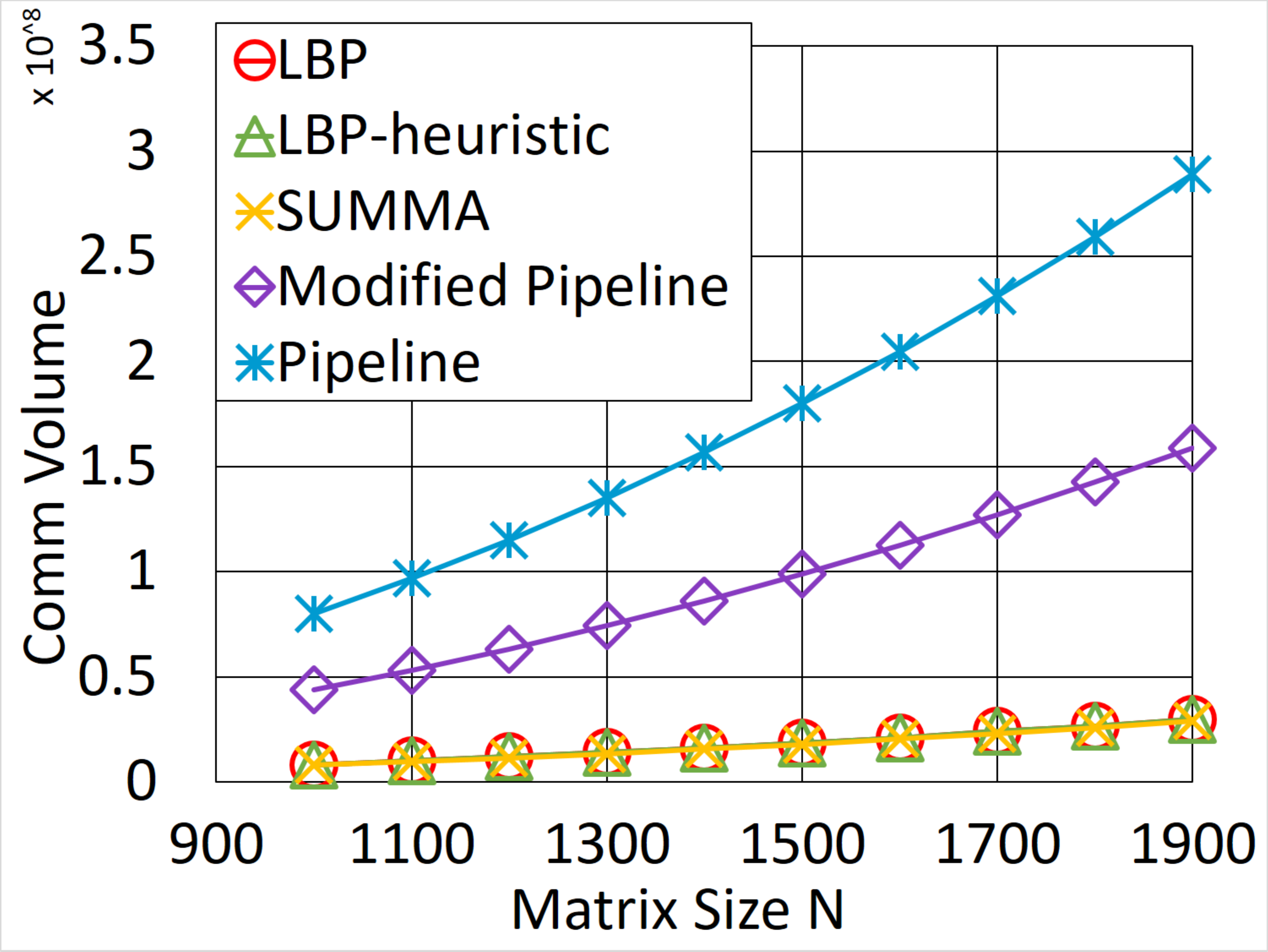}
  \label{fig:mesh5Comvol}
  }
  \subfigure[7*7 Mesh]{
  \includegraphics[width=2.1in]{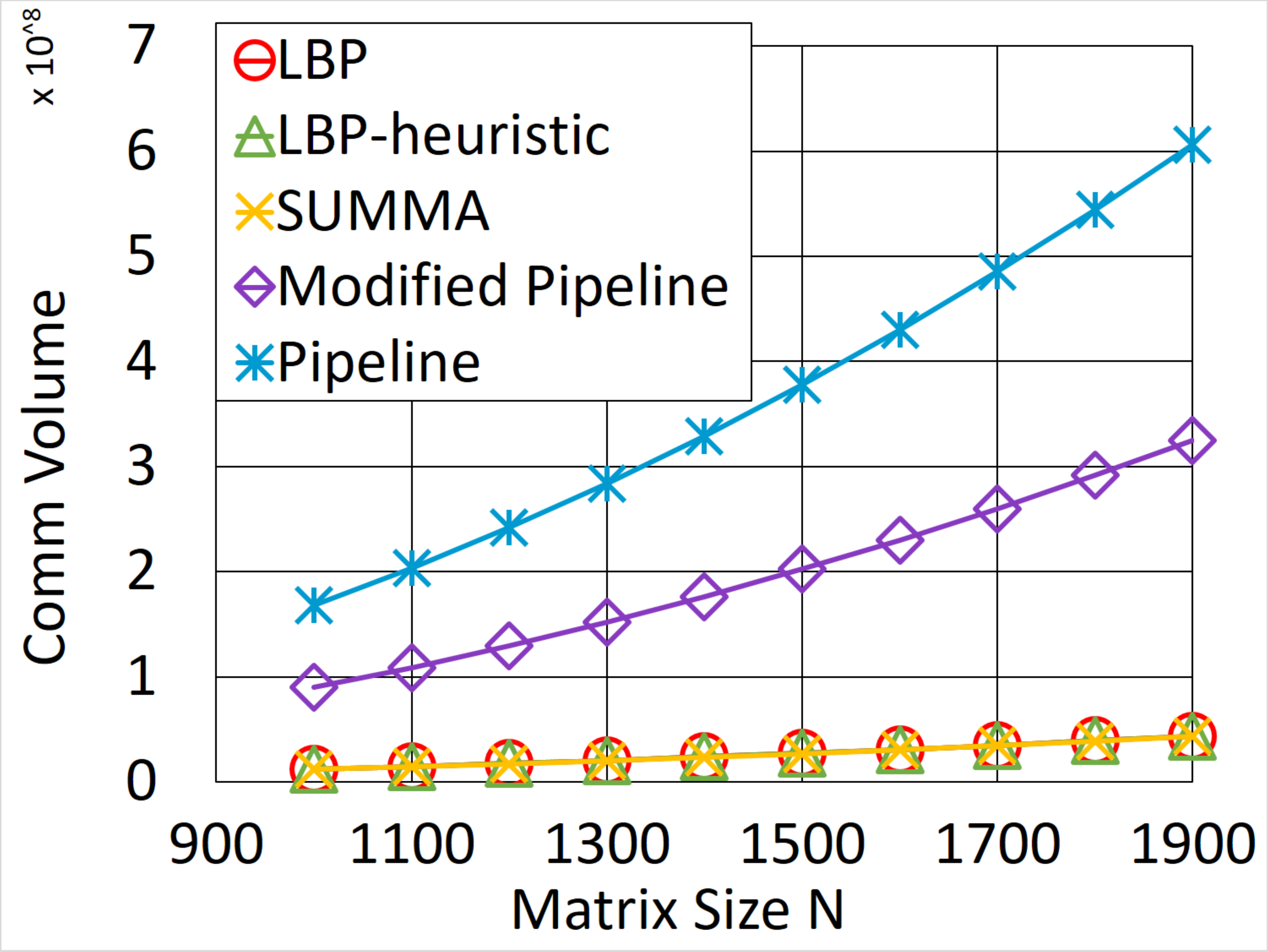}
  \label{fig:mesh7Comvol}
  }
  \subfigure[9*9 Mesh]{
  \includegraphics[width=2.1in]{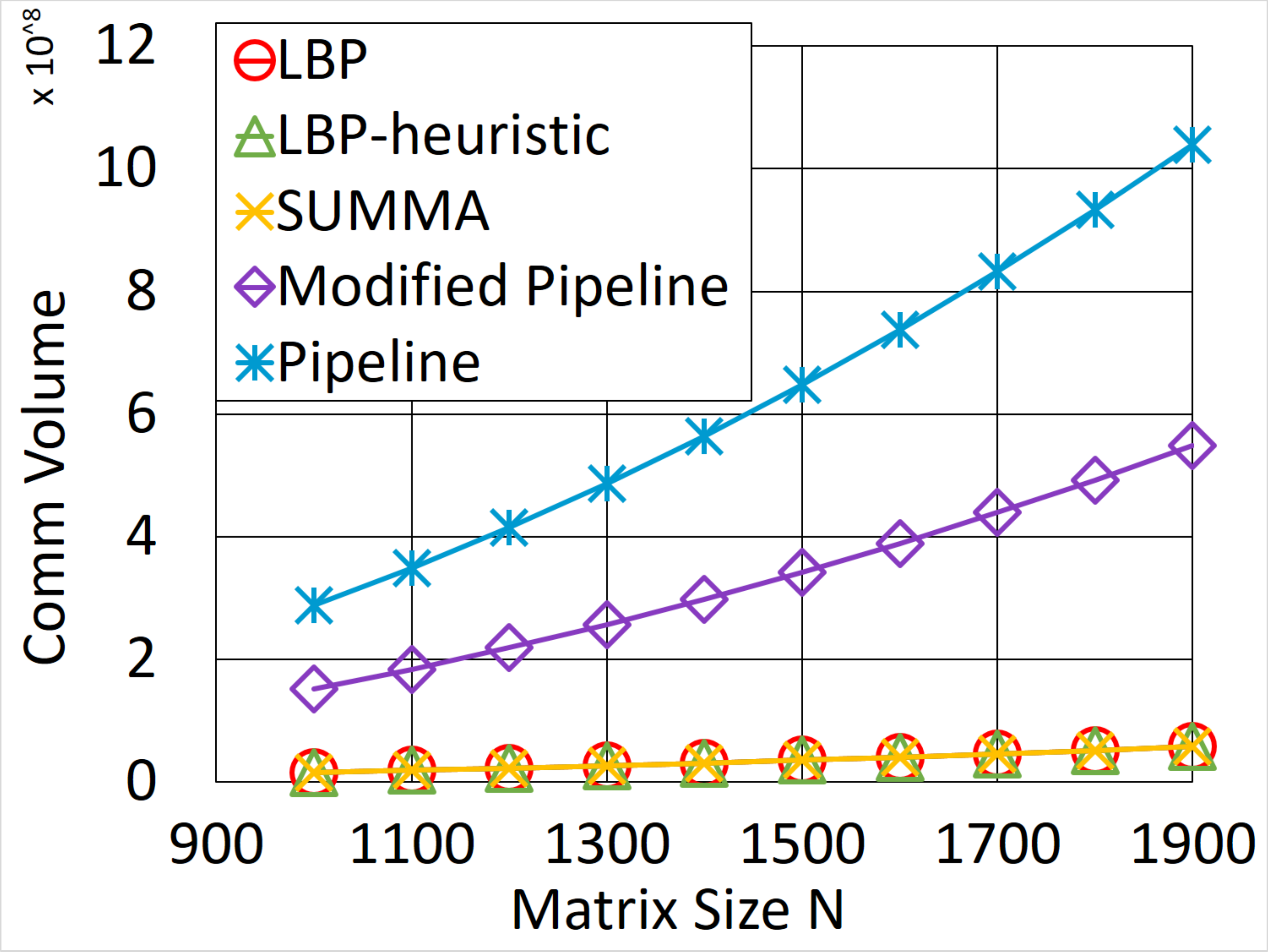}
  \label{fig:mesh9Comvol}
  }
  \vspace{-2mm}
  \caption{Communication volume comparison with increasing matrix size and network dimension.}
  \label{fig:meshComvol}
  \vspace{-5mm}
\end{center}
\end{figure*}

In this subsection, we study how the \emph{layer based partition} scheme performs in mesh networks. In each single run of the simulation, we randomly generate a heterogeneous mesh network, and two square matrices conducting multiplication. Then our LBP algorithm and the other comparing algorithms are called to schedule the matrix multiplication load on the given mesh network.

\noindent
\textbf{Mesh Network.} The mesh network is heterogeneous, with each link speed and processor speed independently generated. The unit processing time $wTcp$ of the processors is uniformly distributed in the range of $(0.0005, 0.0008)$, while the unit transmission time $zTcm$ of the links is uniformly distributed in the range of $(0.0002, 0.0005)$. In our simulation, we use use three square meshes, which are of dimension 5*5, 7*7 and 9*9. For model simplicity and without loss of generality, we focus on studying the case for one quadrant-the lower right one-in figure \ref{fig:meshFlow}, and the source node is located at the top left corner. The cases of the other three quadrants are similar. However, SUMMA is an exception, in which no single source exists, and each processor in the mesh takes one block of matrix data. So when evaluating the performance of SUMMA, we divide the matrix data into blocks and store it on corresponding processor.

\noindent
\textbf{Matrices.} The matrices we analyze are large scale dense matrices. In our simulation, we randomly generate matrices with their side length $N$ ranging from 1000 to 2000.

When analyzing the performance of each algorithm over matrix size, each data point in our simulation is an average of 10 independent experiments over 10 independently different mesh network.

\subsubsection{Evaluation Metrics}
\noindent
\textbf{Overall Communication Volume.} Overall communication volume is defined as the sum of data volume transmitted on each link. Compared to the total data volume coming out of the source, overall communication volume provides a more direct view of data running on each link.

\noindent
\textbf{Task Finishing Time.} The task finishing time in mesh network is defined as the time period from the source starts to send data till the last processor finishes working.

\noindent
\textbf{Total Number of Iterations to Solve LP.} Since we rely on simplex algorithm to solve LP in our LBP and LBP-heuristic, we evaluate both algorithms' efficiency in terms of average total number of iterations taken by simplex algorithm to solve LP.

\subsubsection{Comparison Algorithms}
\noindent
\textbf{SUMMA.} Geijn {\em et al.} \cite{Geijn} propose SUMMA, which is the most widely applied parallel matrix multiplication scheme on homogeneous grid. The algorithm allocates matrix blocks over the grids. In each step, the pivot column of blocks is communicated horizontally and the pivot row of blocks is communicated vertically. Each processor uses the pivot blocks it get to update its rectangle in each step. We apply this algorithm on our heterogeneous mesh network.

\noindent
\textbf{Pipeline.} The Pipeline algorithm is a classic scheduling method. Starting from the source, each node forwards the entire copy of data along the grids to each of its neighbor in the mesh network. Duplicate copies may be sent to one node, however, it only keep the first received one. Once that node finishes receiving its first copy, it starts processing the data while forwarding. The whole system acts like a pipeline with communication overlaps with computing.

\noindent
\textbf{Modified Pipeline.} Tan {\em et al.} \cite{Tan} propose an improved pipeline broadcast scheme for distributed matrix multiplication. The non-blocking pipeline scheme takes advantage of tuned chunk size to boost communication performance. We apply the idea to heterogeneous mesh network.

\subsubsection{Evaluation Result on Mesh}

\noindent
\textbf{Overall Communication Volume.} Figure \ref{fig:meshComvol} displays the overall communication volume when conducting matrix multiplication of two $N*N$ matrices in (a) 5*5 mesh, (b) 7*7 mesh, and (c) 9*9 mesh respectively. The simulation result reveals that while all algorithms's overall communication volume goes up as matrix size increases, SUMMA, LBP and LBP-heuristic generate almost the same smallest overall communication volume, which is $81\%$ smaller than that of Modified Pipeline and $90\%$ smaller than that of Pipeline. SUMMA is well-known to be communication-optimal on homogeneous mesh network. Though applying SUMMA on heterogeneous mesh may affect its overall finishing time due to the change of processor speed and link speed, its data transmission pattern won't be affected. So SUMMA is still communication-optimal on heterogeneous mesh. LBP and LBP-heuristic generates almost the same communication volume as SUMMA, consequently, implies that LBP and LBP-heuristic are at least close to communication optimal on heterogeneous mesh. Moreover, we observe that LBP, LBP-heuristic, SUMMA are close to each other as network dimension increases, but their difference ratio with the other algorithms are getting larger and larger.

\noindent
\textbf{Task Finishing Time.} Figure \ref{fig:meshFinishTime} shows the task finishing time of each algorithm on (a) 5*5 mesh, (b) 7*7 mesh, and (c) 9*9 mesh. Generally, LBP generates the smallest task finishing time than the rest of algorithms. LBP-heuristic gives a task finishing time that is slightly longer than that of LBP, which are $0.03\%$ more in 5*5 mesh, $0.08\%$ more in 7*7 mesh, and $0.18\%$ more in 9*9 mesh, respectively. This tiny difference can entirely be ignored. SUMMA, since it can no longer reach load balance with link speed and processor speed vary, its task finishing time are, respectively, $56.4\%$ more in 5*5 mesh, $52.9\%$ more in 7*7 mesh, and $46.7\%$ more in 9*9 mesh, than that of LBP. Moreover, Modified Pipeline are respectively $66.7\%$ more in 5*5 mesh, $87.2\%$ more in 7*7 mesh, and $121.1\%$ more in 9*9 mesh. Pipeline are respectively $73.4\%$ more in 5*5 mesh, $114\%$ more in 7*7 mesh, and $185\%$ more in 9*9 mesh. All in all, LBP and LBP-heuristic present the best performance in terms of task finishing time.

\begin{figure*}
\begin{center}
  \subfigure[5*5 Mesh]{
  \includegraphics[width=2.1in]{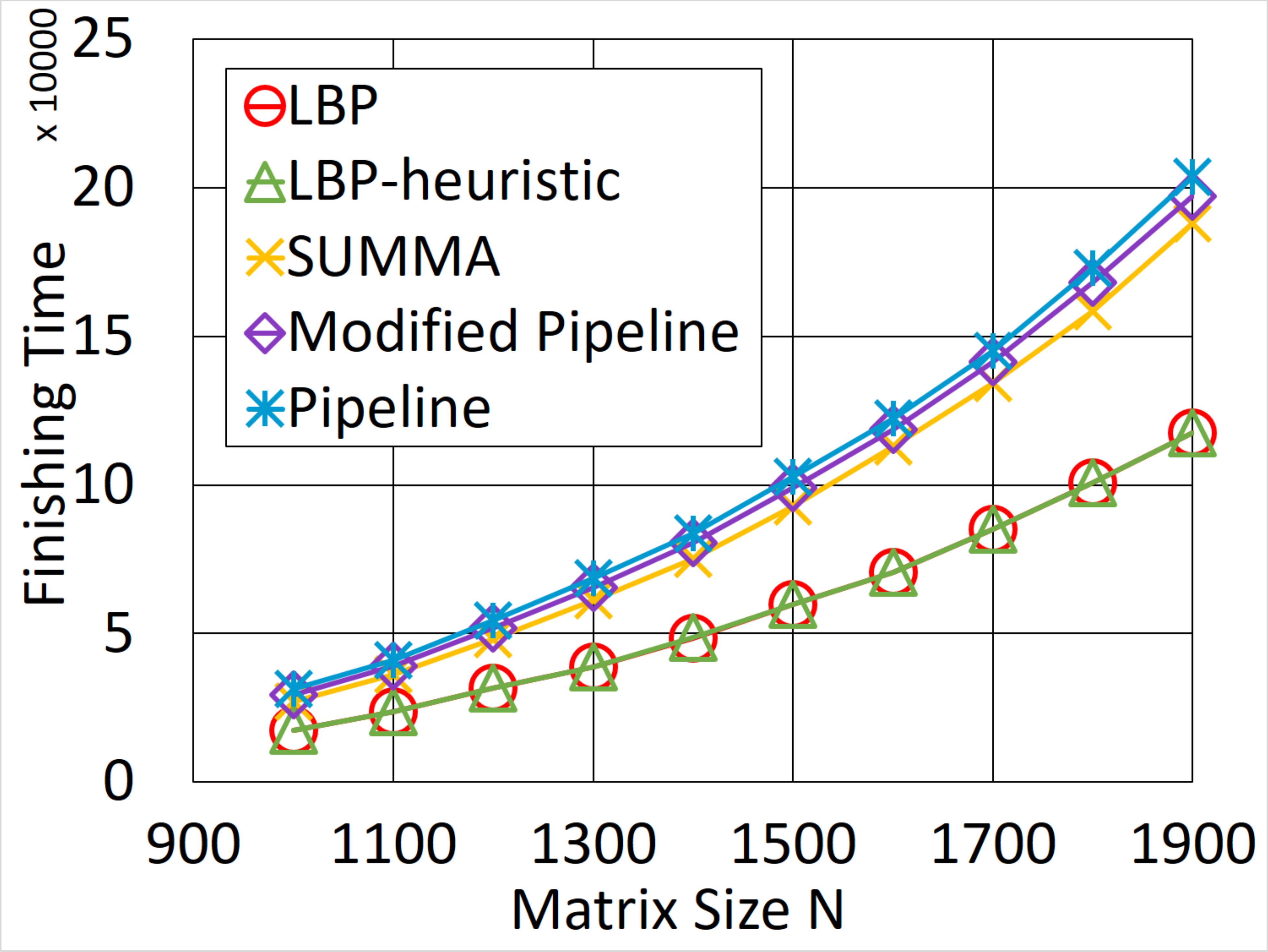}
  \label{fig:mesh5FinishTime}
  }
  \subfigure[7*7 Mesh]{
  \includegraphics[width=2.1in]{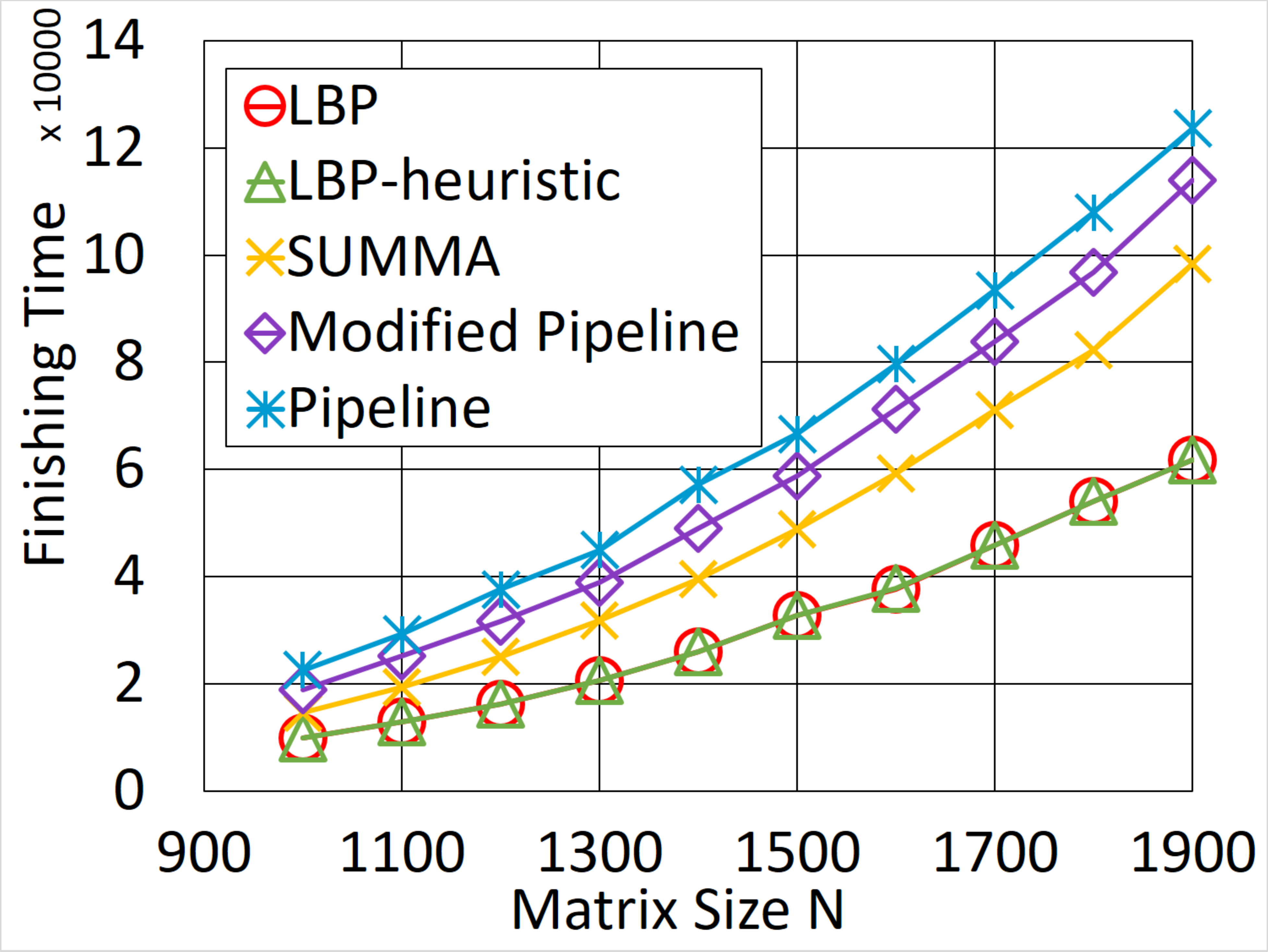}
  \label{fig:mesh7FinishTime}
  }
  \subfigure[9*9 Mesh]{
  \includegraphics[width=2.1in]{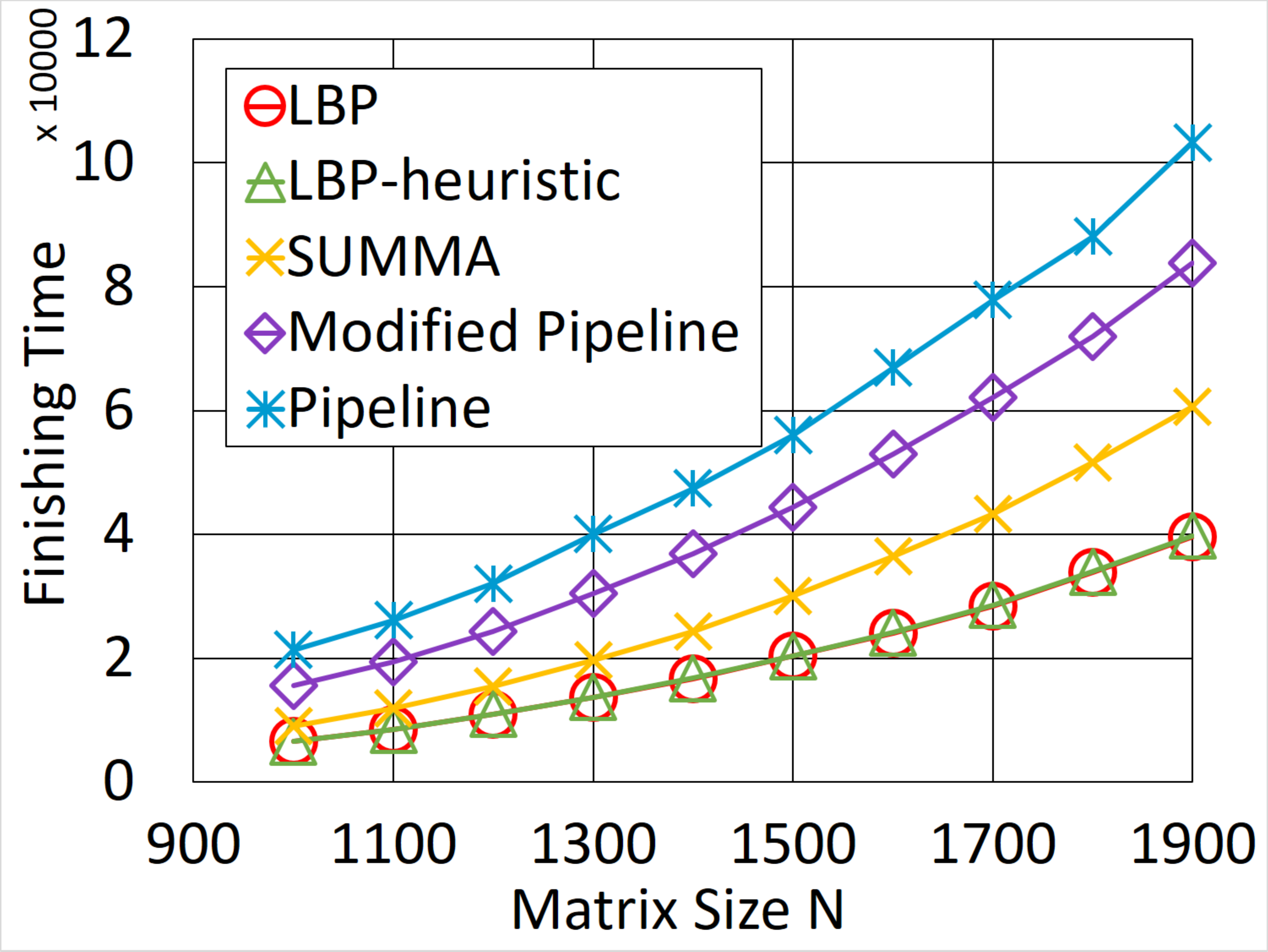}
  \label{fig:mesh9FinishTime}
  }
  \vspace{-2mm}
  \caption{Finishing time comparison with increasing matrix size and network dimension.}
  \label{fig:meshFinishTime}
  \vspace{-5mm}
\end{center}
\end{figure*}

\noindent
\textbf{Total Number of Iterations to Solve LP.} As mentioned previously, we use the simplex algorithm to solve LP in our LBP and LBP-heuristic algorithm. Each time solving the LP costs a certain number of iterations by the simplex algorithm. And according to our algorithm, we may re-solve LP a couple of times due to 1.find real number solution 2. find integer solution 3. local search, etc. Therefore, the total number of iterations to solve LP is a good indication of the efficiency of our algorithms. Figure \ref{fig:simplex} counts the average total number of iterations in solving LP by LBP and LBP-heuristic on 5*5 mesh, 7*7 mesh, 9*9 mesh. Each point is an average of 10 identical independent experiments. The solid lines represent LBP whereas the dashed lines represent LBP-heuristic. We have the following observations: 1. The solid lines vary dramatically due to the uncertain number of times to re-solve LP by LBP, while the dashed lines are comparatively stable. 2. The total number iterations show no correlation with respect to matrix size, a good evidence indicating that both algorithms are suitable for large scale matrix scheduling. 3. The total number of iterations does show positive correlation with respect to mesh size. 4. For the same mesh size, dashed lines are generally far below solid line, which indicates that LBP-heuristic generally requires much less total number of iterations to find a solution than LBP. In other words, LBP-heuristic is more efficient.

\begin{figure}[H]
\centering
  \includegraphics[width=2.1in]{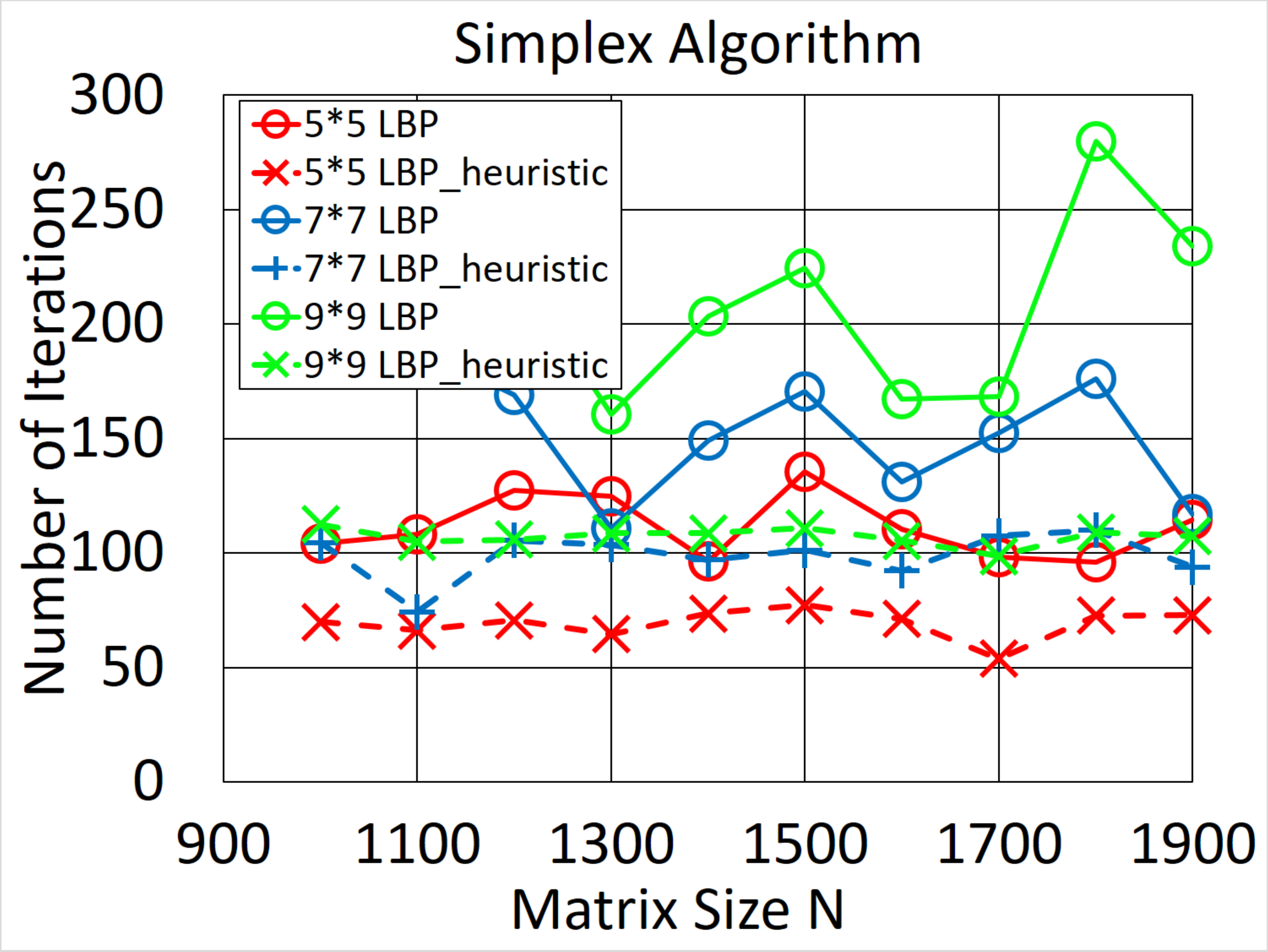}
  \caption{Number of iterations by the simplex method in solving LP.}
  \label{fig:simplex}
%\vspace{-0.2in}
\end{figure}

\noindent
\textbf{Summary.} LBP-heuristic is significantly more efficient than LBP while maintaining almost equally good performance, which makes it widely applicable. Additionally, both algorithms outperform the other heterogeneous mesh scheduling algorithms substantially.

\section{Conclusion}
\label{conclusion}
In this paper, we focus on the problem of scheduling matrix multiplication on heterogeneous processor platforms. We first address two drawbacks of traditional \emph{rectangular partition}: 1. difficult to determine partition shapes. 2. communication volume is not optimal. Alternatively, we present a novel scheduling method: \emph{layer based partition}. We demonstrate that \emph{layer based partition} scheme is easy to find a partition, and generates a total communication volume that is optimal, smaller than the lower bound of \emph{rectangular partition}. In the following part, we study how to minimize task finishing time using LBP. In \emph{single-neighbor network}, we propose an equality based theory. In \emph{multi-neighbor network}, we formulate the problem as a Mix Integer Programming problem, which we provide a 3-Phase algorithm to solve. Considering the high time complexity, a heuristic algorithm is also proposed.

Simulation results show that \emph{layer based partition} outperforms the other comparing algorithms both in single and multiple neighbor networks. It generates a total communication volume that is substantially reduced from the state-of-the-art \emph{rectangular partition} algorithms, and maintain load balancing so that its task finishing time is minimized as well. Results also shows that in mesh network, LBP-heuristic achieves close, and even equal performance to \emph{layer based partition}, with its total number of iterations significantly reduced. Hence, we believe LBP-heuristic is very applicable in real world practice.

%\begin{acknowledgements}
%\end{acknowledgements}



\end{document}